\documentclass[lettersize,journal]{IEEEtran}
\usepackage{cite}
\usepackage{amsmath,amssymb,amsfonts}
\usepackage{xcolor}
\usepackage{mathrsfs} 
\usepackage{siunitx}
\usepackage{mathtools}
\usepackage{algorithmic}
\usepackage{graphicx}
\usepackage{textcomp}
\usepackage{subfigure}
\usepackage{epstopdf}
\usepackage{multirow}%
\usepackage{colortbl}
\graphicspath{{Figure/}}
\usepackage{booktabs}
\usepackage{caption}
\usepackage{wrapfig}
\usepackage{picinpar}
\usepackage{url}
\usepackage{flushend}
\usepackage[latin1]{inputenc}
\usepackage{colortbl}
\usepackage{soul}
\usepackage{multirow}
\usepackage{pifont}
\usepackage{color}
\usepackage{alltt}
\usepackage{enumerate}
\usepackage{siunitx}
\usepackage{breakurl}
\usepackage{pbox}
\usepackage{mathrsfs}

\newtheorem{proof}{Proof} %[section]
\newtheorem{theorem}{Theorem}
\newtheorem{assumption}{Assumption}

\newtheorem{lemma}{Lemma}
\newtheorem{definition}{Definition}

\begin{document}

\title{Enhancing Safety in Nonlinear Systems: Design and Stability Analysis of Adaptive Cruise Control}
\author{Fan Yang, Haoqi Li, Maolong Lv,  Jiangping Hu, \IEEEmembership{Senior Member, IEEE}, Qingrui Zhou, and Bijoy K. Ghosh, \IEEEmembership{Life Fellow, IEEE}
	
	\thanks{This work was partially supported by the National Key Research and Development Program of China under Grant 2022YFE0133100, the State Key Program of National Natural Science Foundation of China under Grant U21B2008, and by the Sichuan Science and Technology Program under Grant 24NSFSC1362. \emph{(Corresponding author: Jiangping Hu.)}}
	\thanks{F. Yang, H. Li and J. Hu are with the School of Automation Engineering, University of Electronic Science and Technology of China, Chengdu 611731, China  e-mail: (yanf@std.uestc.edu.cn; lihq@std.uestc.edu.cn; hujp@uestc.edu.cn).}
	\thanks{M. Lv is with Air Force Engineering University, Xian, China (e-mail: m.lyu@tudelft.nl).}
	\thanks{Q. Zhou is with the Qianxuesen Lab, China Academy of Space Technology, Beijing, China (e-mail: zhouqingrui@spacechina.com).}
	\thanks{B. K. Ghosh is with the Department of Mathematics and Statistics, Texas Tech University, Lubbock, TX, 79409-1042, USA, and also with the School of Automation Engineering, University of Electronic Science and Technology of China, Chengdu 611731, China (e-mail: bijoy.ghosh@ttu.edu).}
	
}

\maketitle

\begin{abstract}
	The safety of autonomous driving systems, particularly self-driving vehicles, remains of paramount concern. These systems exhibit affine nonlinear dynamics and face the challenge of executing predefined control tasks while adhering to state and input constraints to mitigate risks. However, achieving safety control within the framework of control input constraints, such as collision avoidance and maintaining system states within secure boundaries, presents challenges due to limited options. In this study, we introduce a novel approach to address safety concerns by transforming safety conditions into control constraints with a relative degree of 1. This transformation is facilitated through the design of control barrier functions, enabling the creation of a safety control system for affine nonlinear networks. Subsequently, we formulate a robust control strategy that incorporates safety protocols and conduct a comprehensive analysis of its stability and reliability. To illustrate the effectiveness of our approach, we apply it to a specific problem involving adaptive cruise control. Through simulations, we validate the efficiency of our model in ensuring safety without compromising control performance. Our approach signifies significant progress in the field, providing a practical solution to enhance safety for autonomous driving systems operating within the context of affine nonlinear dynamics.
\end{abstract}


\section{Introduction}
\label{sec:introduction}
\IEEEPARstart{I}{n} the realm of autonomous driving systems, especially self-driving vehicles, a variety of approaches have been employed to ensure safety. These methodologies encompass the utilization of safety-resilient event-triggered output feedback control to enhance the vehicle's robustness against attacks, preventing attacks from inducing instability in the vehicle \cite{ding2022security}. Additionally, measures such as providing vehicle motion control for smooth driving conditions and comfortable riding trajectories have been implemented to ensure the overall safety of the vehicle \cite{wei2019risk}. Recently, the organic integration of vehicle motion control and deep learning has further endowed the capability of yielding to pedestrians \cite{crosato2022interaction}. Adaptive cruise control, as a form of autonomous driving technology, faces challenges related to passenger comfort and collision avoidance. To address these challenges, a collaborative design framework has been proposed, eliminating conflicts between comfort and hardware by specifying performance constraints \cite{zhang2022collaborative}. Additionally, collision risks are reduced by carefully constraining the system state within predefined ranges \cite{XuGrizzle2017}. \par

Safety is of paramount importance in autonomous driving, and state-constraint control stands out as a primary measure for reducing collision risks.
Barrier Lyapunov functions \cite{Zhu2022}, model predictive control \cite{Brudigam2023}, and control barrier functions (CBFs) \cite{HJWang2022} represent prevailing approaches that prioritize state-constrained control. CBFs surpass the aforementioned methods due to their unique ability to maintain stable system states under limited control inputs, while also achieving computational efficiency. The integration of CBFs into closed-loop control \cite{TaylorAmes2020} offers safety assurances for system protection. As a result, various domains such as multirobot systems \cite{Pickem2017}, automotive systems \cite{XuWaters2017}, and quadrotors \cite{WangAmes2017ICRA} have successfully harnessed CBFs to enhance safety. However, safety-critical systems require simultaneous consideration of multiple safety conditions, including collision avoidance from all directions and achieving diverse objectives like precise tracking of preceding vehicles and overall system stability \cite{ZhangY2022}. Although CBFs ensure safety, they often fall short in addressing system stability. To address this limitation, the intelligent incorporation of control Lyapunov functions (CLFs) \cite{AmesGalloway2014} with CBFs introduces quadratic programming (QP). This facilitates the development of an optimized controller that prioritizes both safety and control goals \cite{Gangopadhyay2022,HuWang2021}. In applications such as adaptive cruise control, time-headway safety requirements serve as control barrier functions, while stability objectives are defined by control Lyapunov functions. This approach enables collision avoidance while dynamically adjusting vehicle speed in response to acceleration changes or lane shifts by leading vehicles \cite{AmesXu2017}. Despite these advancements, it remains crucial to consider the relationship between optimal control input and constrained inputs, as the obtained optimal control input may not always be feasible due to system limitations. This issue renders methods like CBF-QP and CLF-CBF-QP untenable, as demonstrated in \cite{AmesXu2017}. In addressing this challenge, researchers propose optimizing the lower-bound decay rate in CBFs using variables and introducing an optimal decay variable. This contributes to the feasibility of CLF-CBF-QP under control input constraints \cite{ZengCBF2021}. Importantly, the incorporation of control input ($u$) into CBF constraints is pivotal to the application of the aforementioned techniques. The CLF-CBF-QP approach, incorporating control input limitations, emerges as a practical safety control solution for simpler systems discussed in \cite{AmesXu2017,ZengCBF2021}.
\par

Addressing the application of control barrier functions to safety-critical systems with high-relative degree dynamics presents a formidable challenge. Incorporating the control input $u$ into CBF constraints becomes particularly intricate when safety conditions have a relative degree of at least 2 \cite{Khalil2002}. This is evident in systems like Cartesian robotic dynamics \cite{HJWang2022} or the adaptive cruise control system, which features a relative degree 2 safety constraint \cite{Xiao2019}. In such demanding scenarios, existing methods prove inadequate. To tackle this issue, a high-order control barrier function (HOCBF) was introduced \cite{Xiao2019}, aiming to address time-headway safety conditions and reconcile conflicts between the HOCBF and control input constraints through class $\mathcal{K}$ function penalization. However, this HOCBF-based approach lacks a systematic explanation of the control strategy's feasibility, a crucial aspect for real-world scenarios involving external disturbances and system limitations \cite{Hussain2018}. Input saturation, often the most prevalent input constraint, can significantly impact closed-loop system performance and efficiency, potentially compromising stability in extreme cases \cite{NguyenLaurain2016}. While prior research has addressed diverse issues linked to state and input constraints \cite{Zhang2021,Nguyen2021}, the formulation of an appropriate CBF applicable to general affine systems remains a developing field. Ensuring the feasibility of the CLF-HOCBF-QP method proposed in \cite{Xiao2019}, as demonstrated in \cite{Xiao2022}, involves enforcing a single constrained CBF. Nevertheless, the quest for a suitable CBF applicable to general affine systems is still ongoing. This underscores the significant challenge in guaranteeing the safety of affine nonlinear systems featuring arbitrarily high relative degree safety conditions while adhering to input and state constraints. This challenge serves as one of the primary motivations driving the present study.
\par 

This current study aims to target a category of complex affine nonlinear systems characterized by high-relative degree dynamics and subject to both input and state constraints. The research program's significant contribution lies in the establishment of essential and comprehensive terminology, ensuring the feasibility of quadratic programs through the introduction of a pioneering control barrier function. The central focus is on providing necessary and sufficient conditions for optimizing QPs while working within highly dynamic, nonlinear environments to facilitate efficient and effective closed-loop system performance. The key contributions of this study are as follows: \par 

\begin{itemize}
	\item We introduce an innovative control barrier function, skillfully designed to convert high relative degree safety conditions into equivalent control constraints with a relative degree of 1. This pioneering approach ensures system safety by applying these transformed constraints, all with a consistent relative degree of 1.
	\item We present a safety-stability control strategy for general affine nonlinear systems, founded on a control barrier function proposed in Section \ref{section3}. This strategy enables the simultaneous consideration of two potentially conflicting constraints and facilitates the derivation of both necessary and sufficient conditions to ensure the feasibility of satisfying both constraints. Through the implementation of this control strategy, we seamlessly harmonize safety and stability goals within the context of general affine nonlinear systems.
	\item Our study outlines a robust safety-stability control strategy that can be formulated as inequality-constrained quadratic programs, subsequently solvable through a sequential optimization approach. This versatile control strategy finds application across a wide range of systems, including adaptive cruise control systems. The efficacy of this approach in various control contexts highlights its adaptability and versatility, confirming its potential to address different control scenarios.
\end{itemize} \par 

The subsequent sections of this paper are organized as follows. Section \ref{section2} delves into the fundamental background and sets out the problem formulation. In Section \ref{section3}, we lay out the feasible design for safe control, detailing our innovative approach. The application of our newly-conceived control strategy to the adaptive cruise control (ACC) system is expounded in Section \ref{section4}. In Section \ref{section5}, we explicate the simulation results obtained from implementing our novel methodology to the ACC system and affirm its efficacy. Ultimately, in Section \ref{section6}, we conclude by highlighting the salient outcomes and contributions derived throughout our research endeavor. \par   

\section{Preliminaries and Problem Formulation}\label{section2}
\subsection{Preliminaries}
Study an affine control system, which has the following form
\begin{equation}\label{1}
\dot{\boldsymbol{\zeta}}=\mathtt{\hat{F}}(\boldsymbol{\zeta})+\mathtt{\hat{G}}(\boldsymbol{\zeta})\boldsymbol{\mathfrak{\upsilon}},
\end{equation}
where $\boldsymbol{\zeta}=\left[\mathtt{\zeta}_1, \mathtt{\zeta}_2, \ldots, \mathtt{\zeta}_n\right]^T \in \mathtt{X} \subset \mathbb{R}^\mathtt{n}$. The state constraint set of system \eqref{1} is expressed as
\begin{equation}\label{3}
\mathtt{\bar{C}}:=\left\{\boldsymbol{\zeta} \in \mathbb{R}^\mathtt{n}: \mathtt{\theta}(\boldsymbol{\zeta}(t)) \geq 0\right\},
\end{equation}
$\mathtt{\theta}(\boldsymbol{\zeta}(t)) \subset \mathbb{R}$ is a continuous differentiable bounded function, $\mathtt{\hat{F}}(\boldsymbol{\zeta})=\left[\mathtt{\hat{F}}_\mathtt{1}(\boldsymbol{\zeta}), \mathtt{\hat{F}}_\mathtt{2}(\boldsymbol{\zeta}), \ldots, \mathtt{\hat{F}}_\mathtt{n}(\boldsymbol{\zeta})\right]^T \subset \mathbb{R}^\mathtt{n}$ and $ \mathtt{\hat{G}}(\boldsymbol{\zeta})=\left[\mathtt{\hat{G}}_\mathtt{1}(\boldsymbol{\zeta}), \mathtt{\hat{G}}_\mathtt{2}(\boldsymbol{\zeta}), \ldots, \mathtt{\hat{G}}_\mathtt{n}(\boldsymbol{\zeta})\right]^T \subset \mathbb{R}^{\mathtt{n} \times \mathtt{q} }$ exhibit local Lipschitz continuity. The set of control constraints $\boldsymbol{\mathfrak{\upsilon}}\in \mathtt{\Upsilon} \subset{\mathbb{R}^\mathtt{q}}$ is designated as follows \par
\begin{equation}\label{2}
\mathtt{\Upsilon}:=\{\boldsymbol{\mathfrak{\upsilon}} \subset {\mathbb{R}^\mathtt{q}}:\boldsymbol{\mathfrak{\upsilon}}_{\min} \leq \boldsymbol{\mathfrak{\upsilon}}\leq \boldsymbol{\mathfrak{\upsilon}}_{\max}\},
\end{equation}
where $\boldsymbol{\mathfrak{\upsilon}}_{\max},\boldsymbol{\mathfrak{\upsilon}}_{\min} \subset {\mathbb{R}^\mathtt{q}}$ and \eqref{2} is understood by considering each component separately. \par 

\begin{definition}\label{definition5}
	\cite{Glotfelter2017} A set $\mathtt{\bar{C}}\subset \mathbb{R}^\mathtt{n}$  is considered forward invariant for system \eqref{1} if all trajectories originating from an initial state $\boldsymbol{\zeta}({t_\mathtt{0}})\in \mathtt{\bar{C}}$ remain confined within $\mathtt{\bar{C}}$ for all $t$ greater than or equal to $t_\mathtt{0}$.  
\end{definition}

\begin{definition}\label{definition2}
	\cite{AmesXu2017} The continuously differentiable function $\mathtt{\theta}(\boldsymbol{\zeta}(t))$ is a barrier function (BF) for system \eqref{1} if there exists a class $\mathcal{K}$ function $\alpha$ such that
	\begin{equation}\label{4}
	\dot{\mathtt{\theta}}(\boldsymbol{\zeta})+\alpha({\mathtt{\theta}}(\boldsymbol{\zeta})) \geq 0, \forall \boldsymbol{\zeta} \in \mathtt{\bar{C}}.
	\end{equation}
\end{definition} 

\begin{lemma}\label{lemma1}
	\cite{AmesXu2017} If there exists a BF ${\mathtt{\theta}(\boldsymbol{\zeta}(t))}:\mathtt{\bar{C}} \to \mathbb{R}$, then $\mathtt{\bar{C}}$ is forward invariant for \eqref{1}.
\end{lemma}

\begin{definition}\label{definition6}
	\cite{Lindemann2019} ${\mathtt{\theta}(\boldsymbol{\zeta}(t))}$ is a control barrier function (CBF) for system \eqref{1} if there exists a class $\mathcal{K}$ function $\alpha$ such that 
	\begin{equation}\label{5}
	\sup _{\boldsymbol{\mathfrak{\upsilon}} \in \mathtt{\Upsilon}}\left[L_{\mathtt{\hat{F}}} \mathtt{\theta}(\boldsymbol{\zeta})+L_{\mathtt{\hat{G}}} \mathtt{\theta}(\boldsymbol{\zeta}) \boldsymbol{\mathfrak{\upsilon}}+\alpha(\mathtt{\theta}(\boldsymbol{\zeta}))\right] \geq 0,  \forall \boldsymbol{\zeta} \in \mathtt{\bar{C}}.
	\end{equation}
\end{definition} 
\noindent Here, $L_{\mathtt{\hat{F}}}$ and $L_{\mathtt{\hat{G}}}$ denote the Lie derivatives along $\mathtt{\hat{F}}(\boldsymbol{\zeta})$ and $\mathtt{\hat{G}}(\boldsymbol{\zeta})$, respectively.

\begin{definition}\label{definition4}
	\cite{AmesGalloway2014} A continuously differentiable function $\mathtt{V}(\boldsymbol{\zeta}) \subset \mathbb{R}$ can be regarded as a control Lyapunov function (CLF) that ensures exponential stability of system \eqref{1} if there exist constants  $\mathtt{\chi}_1 ~ \textgreater ~ 0,\mathtt{\chi}_2 ~ \textgreater ~ 0,\mathtt{\chi}_3 ~ \textgreater ~ 0$  such that $\forall \boldsymbol{\zeta}\in \mathtt{X}$, $\mathtt{\chi}_{1}\|\boldsymbol{\zeta}\|^{2} \leq \mathtt{V}(\boldsymbol{\zeta}) \leq \mathtt{\chi}_{2}\|\boldsymbol{\zeta}\|^{2}$,
	\begin{equation}\label{6}
	\inf _{\boldsymbol{\mathfrak{\upsilon}} \in \mathtt{\Upsilon}}\left[L_{\mathtt{\hat{F}}} \mathtt{V}(\boldsymbol{\zeta})+L_{\mathtt{\hat{G}}} \mathtt{V}(\boldsymbol{\zeta}) \boldsymbol{\mathfrak{\upsilon}}+\mathtt{\chi}_{3} \mathtt{V}(\boldsymbol{\zeta})\right] \leq 0.
	\end{equation}
\end{definition}

\begin{definition}\label{definition3}
	\cite{Khalil2002} The Relative Degree of $\mathtt{\theta}(\boldsymbol{\zeta})$ in relation to system \eqref{1} can be characterized as the minimal number of differentiations required along the dynamics of \eqref{1} until the distinct influence of control $\boldsymbol{\mathfrak{\upsilon}}$ is clearly revealed.
\end{definition}

\begin{assumption}\label{assumption1}
	The relative degree of security constraints $\mathtt{\theta}(\boldsymbol{\zeta})$ is ${m (m\geq2)}$.
\end{assumption}

For system \eqref{1}, we define $G_d(\boldsymbol{\zeta}):=\left(\boldsymbol{\zeta}+\boldsymbol{\mathfrak{d}}\right)^T \mathtt{\hat{G}}(\boldsymbol{\zeta})$, where $\boldsymbol{\mathfrak{d}}=\left[d_1, d_2, \ldots, d_n\right]^T \subset \mathbb{R}^\mathtt{n}$ is a constant vector.

\begin{assumption}\label{assumption2}
	$G_d(\boldsymbol{\zeta}) \neq \boldsymbol{0}, \forall t \geq t_\mathtt{0}$.
\end{assumption} \par

\subsection{Problem Formulation} \label{PF}
Considering system \eqref{1} and assuming Assumptions \ref{assumption1} and \ref{assumption2} are satisfied, our objective is to design a CBF with a relative degree of 1. It will provide a practical control approach for system \eqref{1} and ensure the system's state trajectory closely tracks the desired state. \par 
Consider system (1), with the following constraints \par 
\noindent \textbf{Input constraint}: $\boldsymbol{\mathfrak{\upsilon}}_{\min} \leq \boldsymbol{\mathfrak{\upsilon}}\leq \boldsymbol{\mathfrak{\upsilon}}_{\max}$; \\
\textbf{State restriction}: $\mathtt{\theta}(\boldsymbol{\zeta}(t)) \geq 0$;  \par 
We aim to devise a CBF $\mathtt{\Theta}(\boldsymbol{\zeta}(t))$ that conforms to the input and state limitations while simultaneously diminishing the tracking error $\left\|\boldsymbol{\zeta}(t)-\boldsymbol{\zeta}_\mathtt{d}\right\| \leq \epsilon$, where $\boldsymbol{\zeta}_\mathtt{d} \in \mathbb{R}^\mathtt{n}$ is a fixed constant state and $\epsilon$ is a positive constant.

\section{Feasible Design of Safe Control}\label{section3}
A time-varying function is formulated to define an invariant set for system \eqref{1} 
\begin{equation}\label{7}
\mathtt{\Theta}(\boldsymbol{\zeta}(t)) := e^{\frac{\mathtt{\theta}(\boldsymbol{\zeta}(t))}{\|\boldsymbol{\zeta}+\boldsymbol{\mathfrak{d}}\|+\mathtt{r}}-\mathtt{\Delta}}-1,
\end{equation}
where ${\mathtt{\Theta}(\boldsymbol{\zeta}(t)) \subset \mathbb{R}}$ and $\mathtt{\Delta} \in (0, +\infty),  \mathtt{r} \in (0, +\infty)$. \par 

\begin{lemma}\label{lemma2}
	Under Assumptions \ref{assumption1} and \ref{assumption2}, the relative degree of $\mathtt{\Theta}(\boldsymbol{\zeta})$ is 1, and $\mathtt{\bar{C}}_{\mathtt{\Theta}}:=\left\{\boldsymbol{\zeta} \in \mathbb{R}^\mathtt{n}: \mathtt{\Theta}(\boldsymbol{\zeta}) \geq 0\right\}$ is forward-invariant.
\end{lemma}

\subsection{Method for Addressing Input and State Constraints While Minimizing Tracking Error} 
Prior to devising a secure control blueprint, it is crucial to explore an approach that guarantees non-interference between input constraints and state limitations whilst simultaneously minimizing tracking discrepancies. In accordance with the input constraint specifications, the control law must adhere to the following stipulations
\begin{equation}\label{45}
\boldsymbol{\mathfrak{\upsilon}}_{\min} \leq \boldsymbol{\mathfrak{\upsilon}}\leq \boldsymbol{\mathfrak{\upsilon}}_{\max}.
\end{equation} \par 
From Lemma \ref{lemma2}, it is easier to ensure $\mathtt{\Theta}(\boldsymbol{\zeta})\geq 0$ than $\mathtt{\theta}(\boldsymbol{\zeta})\geq 0$. Therefore, we can employ a nonlinear control barrier function (NCBF) \eqref{7} to enforce the state restriction
\begin{equation}\label{44}
L_{\mathtt{\hat{F}}} \mathtt{\Theta}(\boldsymbol{\zeta})+L_{\mathtt{\hat{G}}} \mathtt{\Theta}(\boldsymbol{\zeta}) \boldsymbol{\mathfrak{\upsilon}}+\alpha(\mathtt{\Theta}(\boldsymbol{\zeta})) \geq 0.
\end{equation}

Ultimately, the trajectory of the system's state must remain predictable within the vicinity of $\boldsymbol{\zeta}_\mathtt{d}$. To achieve this, we can define a continuously differentiable function $\mathtt{V}(\boldsymbol{\zeta}(t))=\left(\boldsymbol{\zeta}(t)-\boldsymbol{\zeta}_{\mathtt{d}}\right)^{T} \boldsymbol{Z} \left(\boldsymbol{\zeta}(t)-\boldsymbol{\zeta}_{\mathtt{d}}\right), \boldsymbol{Z}$ is a positive-definite matrix. For system \eqref{1}, the time derivative of $\mathtt{V}$, denoted as $\dot{\mathtt{V}}$, is given by $\dot{\mathtt{V}} = \frac{\partial \mathtt{V}}{\partial \boldsymbol{\zeta}} \dot{\boldsymbol{\zeta}} = L_{\mathtt{\hat{F}}} \mathtt{V}(\boldsymbol{\zeta}(t))+L_{\mathtt{\hat{G}}} \mathtt{V}(\boldsymbol{\zeta}(t)) \boldsymbol{\mathfrak{\upsilon}}$. Consequently, there exists a $\mathcal{M}$ ($\mathcal{M}$ is a sufficiently large positive number), such that
\begin{equation}\label{47}
L_{\mathtt{\hat{F}}} \mathtt{V}(\boldsymbol{\zeta}(t))+L_{\mathtt{\hat{G}}} \mathtt{V}(\boldsymbol{\zeta}(t)) \boldsymbol{\mathfrak{\upsilon}}+\mathtt{\chi}_3 \mathtt{V}(\boldsymbol{\zeta}(t)) \leq \mathcal{M}.
\end{equation} \par 

The Comparison Lemma \cite{Khalil2002} implies that 
\begin{equation}\label{48}
\mathtt{V}(\boldsymbol{\zeta}(t)) \leq \frac{\mathcal{M}}{\mathtt{\chi}_3} + \left(\mathtt{V}(\boldsymbol{\zeta}(t_\mathtt{0}))-\frac{\mathcal{M}}{\mathtt{\chi}_3}\right)e^{\mathtt{\chi}_3 (-t+t_\mathtt{0})}.
\end{equation} \par 
Based on the definition of $\mathtt{V}(\boldsymbol{\zeta}(t))$, it follows that
\begin{equation}\label{12}
\lambda_{\min}(\boldsymbol{Z}){\left\|\boldsymbol{\zeta}(t)-\boldsymbol{\zeta}_\mathtt{d}\right\|}^{2} \leq  \mathtt{V}(\boldsymbol{\zeta}(t)) \leq \lambda_{\max}(\boldsymbol{Z}){\left\|\boldsymbol{\zeta}(t)-\boldsymbol{\zeta}_\mathtt{d}\right\|}^{2},
\end{equation}
\noindent where $\lambda_{\min}(\boldsymbol{Z})>0$ and $\lambda_{\max}(\boldsymbol{Z})>0$ denote the minimum and maximum eigenvalues of $\boldsymbol{Z}$, respectively. From \eqref{48} and \eqref{12}, we can derive the inequality
\begin{equation}\label{16}
{\left\|\boldsymbol{\zeta}(t)\!-\!\boldsymbol{\zeta}_\mathtt{d}\right\|} \!\leq\! \sqrt{\frac{1}{\lambda_{\min}(\boldsymbol{Z})}\!\!\left[\! \frac{\mathcal{M}}{\mathtt{\chi}_3}\! +\! \!\left(\!\mathtt{V}(\boldsymbol{\zeta}(t_0))\!-\!\frac{\mathcal{M}}{\mathtt{\chi}_3}\right)e^{\mathtt{\chi}_3 (-t+t_\mathtt{0})} \right]}.
\end{equation} \par 
Let $\mathtt{y}(t) = \frac{\mathcal{M}}{\mathtt{\chi}_3} + \left(\mathtt{V}(\boldsymbol{\zeta}(t_{\mathtt{0}})) - \frac{\mathcal{M}}{\mathtt{\chi}_3}\right)e^{\mathtt{\chi}_3 (-t+t_{\mathtt{0}})}$. We observe that $\mathtt{y}(t) = \frac{\mathcal{M}}{\mathtt{\chi}_3}(1-e^{\mathtt{\chi}_3 (-t+t_{\mathtt{0}})}) + \mathtt{V}(\boldsymbol{\zeta}(t_{\mathtt{0}}))e^{\mathtt{\chi}_3 (-t+t_{\mathtt{0}})} \geq 0$ for $\mathtt{V}(\boldsymbol{\zeta}(t_{\mathtt{0}})) \geq 0$. 

Considering $\dot{\mathtt{y}}(t) = -\mathtt{\chi}_3 \left(\mathtt{V}(\boldsymbol{\zeta}(t_{\mathtt{0}})) - \frac{\mathcal{M}}{\mathtt{\chi}_3}\right)e^{\mathtt{\chi}_3 (-t+t_{\mathtt{0}})}$, we conclude that $\mathtt{y}(t)$ can be seen as a strictly monotonic function of $t$. With $\mathtt{y}(t_{\mathtt{0}}) = \mathtt{V}(\boldsymbol{\zeta}(t_{\mathtt{0}}))$ and $\underset{t\rightarrow \infty}{\lim}\mathtt{y}(t) = \frac{\mathcal{M}}{\mathtt{\chi}_3}$, we deduce that $\mathtt{y}_{\max} = \max\{\mathtt{V}(\boldsymbol{\zeta}(t_{\mathtt{0}})), \frac{\mathcal{M}}{\mathtt{\chi}_3}\}$. Let $\epsilon = \max\{\sqrt{\frac{\mathcal{M}}{\lambda_{\min}(\boldsymbol{Z})\mathtt{\chi}_3}}, \sqrt{\frac{\mathtt{V}(\boldsymbol{\zeta}(t_{\mathtt{0}}))}{\lambda_{\min}(\boldsymbol{Z})}}\}$; thus, the system's state trajectory remains within the vicinity of the desired state. \par
We employed a permissive control Lyapunov function to attain the convergence of the tracking error.
\begin{equation}\label{46}
L_{\mathtt{\hat{F}}} \mathtt{V}(\boldsymbol{\zeta}(t))+L_{\mathtt{\hat{G}}} \mathtt{V}(\boldsymbol{\zeta}(t)) \boldsymbol{\mathfrak{\upsilon}}+\mathtt{\chi}_3 \mathtt{V}(\boldsymbol{\zeta}(t)) \leq \delta(\boldsymbol{\zeta}).
\end{equation}
Here, $\delta(\boldsymbol{\zeta})$ denotes the relaxation (decision variable). \par

\subsection{Design of Safety Controller for General Affine Systems} 
In light of the issues elucidated in Section \ref{PF} and with a view towards scalability, optimal control constitutes a suitable modality for addressing the foregoing predicaments. A notable advantage of inequality-constrained optimal control lies in its ability to harmonize control performance --- as signified by a relaxed CLF $\eqref{46}$) --- with the "safe" set to which the trajectory belongs (as manifested through the NCBF constraints $\eqref{44}$), thereby ensuring that the control signal invariably falls within the prescribed control boundary $\eqref{45}$. The optimal control formulation for system \eqref{1} is expounded below \cite{Xiao2022}
\begin{equation}\label{43}
\begin{array}{c}
\underset{\boldsymbol{\mathfrak{\upsilon}}(t), \delta(\boldsymbol{\zeta})}{\min} \int_{t_\mathtt{0}}^{t_\mathtt{z}} \mathcal{E}(\|\boldsymbol{\mathfrak{\upsilon}}(t)\|)+\mathtt{p}\delta^{2}(\boldsymbol{\zeta}) d t \\
\text { s.t. }\eqref{45}, \eqref{44}, \eqref{46},
\end{array}
\end{equation} 
\noindent where $\|\cdot\|$ represents the $2$-norm. Consider $\mathcal{E}(\cdot)$ as the representation of the energy consumption for system \eqref{1}. Therefore, $\mathcal{E}(\cdot) \geq 0$ and $\mathtt{p}>0, t_\mathtt{z} > 0$. \par 
We discretize time and assume that $\boldsymbol{\zeta}(t) = \boldsymbol{\zeta}(t_\mathtt{\kappa})$ remains constant within each sufficiently small time interval $\left[t_{\mathtt{\kappa}}, t_{\mathtt{\kappa}+1}\right)$, where $\mathtt{\kappa}=0, 1, 2, \ldots, t_{\mathtt{0}}=0$.
If $\int_{t_\mathtt{0}}^{t_\mathtt{z}} \mathcal{E}(\|\boldsymbol{\mathfrak{\upsilon}}(t)\|)+\mathtt{p}\delta^{2}(\boldsymbol{\zeta}) d t$ takes a quadratic form in $\boldsymbol{\mathfrak{\upsilon}}$, then \eqref{43} can be reformulated as a sequence of quadratic program (QP) problems within each time interval \cite{Xiao2022}
\begin{equation}\label{51}
\begin{gathered}
\boldsymbol{\mathfrak{\upsilon}}^{*}(t,\boldsymbol{\zeta})=\arg \min _{\boldsymbol{\mathfrak{\upsilon}}(t,\boldsymbol{\zeta})} \frac{1}{2} \boldsymbol{\mathfrak{\upsilon}}(t,\boldsymbol{\zeta})^{T} \boldsymbol{P} \boldsymbol{\mathfrak{\upsilon}}(t,\boldsymbol{\zeta})+\boldsymbol{G}^{T} \boldsymbol{\mathfrak{\upsilon}}(t,\boldsymbol{\zeta}) \\
\text { s.t. } \boldsymbol{A}_{i}^{T}  \boldsymbol{\mathfrak{\upsilon}}(t,\boldsymbol{\zeta}) \leq \mathtt{\theta}_{i},
\end{gathered}
\end{equation}

\noindent where $\boldsymbol{\mathfrak{\upsilon}}(t,\boldsymbol{\zeta}) := \left[\boldsymbol{\mathfrak{\upsilon}}(t), \delta(\boldsymbol{\zeta})\right]^T$. The symbol $\boldsymbol{P}$ corresponds to a positive-definite matrix. $\boldsymbol{A}_{i}$ and $\mathtt{\theta}_{i}$  represent quantities obtained through the conversion of various inequalities into the form $\boldsymbol{A}_{i}^T\boldsymbol{\mathfrak{\upsilon}}\leq \mathtt{\theta}_{i}$, where all vectors are represented as columns. Additionally, $\boldsymbol{G}^{T}$ signifies a row vector obtained by transposing a column vector. \par 
The QP \eqref{51} can be solved using the interior-point method \cite{Nocedal2006}. The corresponding Lagrangian function for \eqref{51} is given by
$\mathtt{L}(\boldsymbol{\mathfrak{\upsilon}}(t,\boldsymbol{\zeta}), \boldsymbol{L})=\frac{1}{2} \boldsymbol{\mathfrak{\upsilon}}(t,\boldsymbol{\zeta})^{T} \boldsymbol{P} \boldsymbol{\mathfrak{\upsilon}}(t,\boldsymbol{\zeta})+\boldsymbol{G}^{T}\boldsymbol{\mathfrak{\upsilon}}(t,\boldsymbol{\zeta})+\sum_{i=1}^{n} \mathtt{l}_{i}\left(\boldsymbol{A}_{i}^{T} \boldsymbol{\mathfrak{\upsilon}}(t,\boldsymbol{\zeta})-\mathtt{\theta}_{i}\right)$.
The Karush--Kuhn--Tucker condition can be formulated as follows
\begin{equation}\label{37}
\left\{\begin{array}{l}
\boldsymbol{P} \boldsymbol{\mathfrak{\upsilon}}(t,\boldsymbol{\zeta})+ \boldsymbol{G} + \boldsymbol{A}^{T} \boldsymbol{L} =\boldsymbol{0} \\
\boldsymbol{A}_{i}^{T} \boldsymbol{\mathfrak{\upsilon}}(t,\boldsymbol{\zeta}) \leq \mathtt{\theta}_{i} \\
\mathtt{l}_{i}\left(\boldsymbol{A}_{i}^T \boldsymbol{\mathfrak{\upsilon}}(t,\boldsymbol{\zeta})-\mathtt{\theta}_{i}\right)= 0 \\
\mathtt{l}_{i} \geq 0, i=1, \ldots, n.
\end{array}\right.
\end{equation} \par 
We introduce the relaxation variable $\mathtt{s}_{i}$, defined as $\mathtt{s}_{i}=\mathtt{\theta}_{i}-\boldsymbol{A}_{i}^{T} \boldsymbol{\mathfrak{\upsilon}}(t,\boldsymbol{\zeta})$, and then rephrase \eqref{37} as
\begin{equation}\label{38}
\left\{\begin{array}{l}
\boldsymbol{P} \boldsymbol{\mathfrak{\upsilon}}(t,\boldsymbol{\zeta})+\boldsymbol{G}+\boldsymbol{A}^{T} \boldsymbol{L}=\boldsymbol{0}, \\
\boldsymbol{A}_{i}^{T} \boldsymbol{\mathfrak{\upsilon}}(t,\boldsymbol{\zeta})+\mathtt{s}_{i} = \mathtt{\theta}_{i}, \\
\mathtt{s}_{i} \mathtt{l}_{i} = 0, i=1, \ldots, n, 	\\
(\boldsymbol{s}, \boldsymbol{L}) \geq 0.
\end{array}\right.
\end{equation}
After the complementarity measure $\mu := \frac{\boldsymbol{s}^{T} \boldsymbol{L}}{n}$ is defined, then we express \eqref{38} in matrix notation
\begin{equation}\label{39}
\boldsymbol{R}(\boldsymbol{\mathfrak{\upsilon}}(t,\boldsymbol{\zeta}), \boldsymbol{s}, \boldsymbol{L}; \sigma \mu)=\left[\begin{array}{c}
\boldsymbol{P} \boldsymbol{\mathfrak{\upsilon}}(t,\boldsymbol{\zeta})+\boldsymbol{G}+ \boldsymbol{A}^{T} \boldsymbol{L} \\
\boldsymbol{A} \boldsymbol{\mathfrak{\upsilon}}(t,\boldsymbol{\zeta})+\boldsymbol{s}-\boldsymbol{\theta} \\
\boldsymbol{\Lambda} \boldsymbol{S}\boldsymbol{e}-\sigma \mu \boldsymbol{e} 
\end{array}\right]=\boldsymbol{0},
\end{equation}
\noindent where 
\begin{equation} \label{40}
\begin{aligned}
\boldsymbol{\Lambda} &\!\!=\!\!\!\left[\begin{array}{cccc}
\!\!\mathtt{l}_{1} &\!\! 0 & \!\!\ldots &\!\! 0 \\
\!\!0 & \!\!\mathtt{l}_{2} & \!\!\ldots & \!\!0 \\
\!\!\vdots & \!\!\vdots & \!\!\ldots & \!\!\vdots \\
\!\!0 & \!\!0 & \!\!\ldots & \!\!\mathtt{l}_{n}
\end{array}\!\!\!\right]\!, \quad \!\!\!\!\!\boldsymbol{A}\!=\!\!\left[\begin{array}{c}
\!\boldsymbol{A}_{1}^{T}\! \\
\!\boldsymbol{A}_{2}^{T}\! \\
\!\vdots\! \\
\!\boldsymbol{A}_{n}^{T}\!
\end{array}\right]\!, \quad \!\!\!\!\!\boldsymbol{L}=\left[\begin{array}{c}
\mathtt{l}_{1} \\
\mathtt{l}_{2} \\
\vdots \\
\mathtt{l}_{n}
\end{array}\right], \\
\boldsymbol{S} &\!\!=\!\!\!\left[\begin{array}{cccc}
\!\!\mathtt{s}_{1} &\!\! 0 & \!\!\ldots &\!\! 0 \\
\!\!0 & \!\!\mathtt{s}_{2} & \!\!\ldots & \!\!0 \\
\!\!\vdots & \!\!\vdots & \!\!\ldots & \!\!\vdots \\
\!\!0 & \!\!0 & \!\!\ldots & \!\!\mathtt{s}_{n}
\end{array}\!\!\!\right]\!\!, \quad \!\!\!\!\!\! \boldsymbol{s}\!=\!\!\left[\begin{array}{c}
\!\!\mathtt{s}_{1}\!\! \\
\!\!\mathtt{s}_{2}\!\! \\
\!\!\vdots\!\! \\
\!\!\mathtt{s}_{n}\!\!
\end{array}\right]\!\!, \quad \!\!\!\!\!\!\boldsymbol{\theta}\!=\!\!\left[\begin{array}{c}
\!\!\mathtt{\theta}_{1}\!\! \\
\!\!\mathtt{\theta}_{2}\!\! \\
\!\!\vdots\!\! \\
\!\!\mathtt{\theta}_{n}\!\!
\end{array}\right]\!\!, \quad \!\!\!\!\!\!\boldsymbol{e}\!=\!\!\left[\begin{array}{c}
\!\!{1}\!\! \\
\!\!{1}\!\! \\
\!\!\vdots\!\! \\
\!\!{1}\!\!
\end{array}\right],
\end{aligned}
\end{equation}
with $\sigma \in [0, 1]$. By fixing $\mu$ and applying Newton's method to solve \eqref{39}, we obtain $\boldsymbol{R}\left(\boldsymbol{\mathfrak{\upsilon}}, \boldsymbol{s}, \boldsymbol{L}; \sigma \mu\right)+\boldsymbol{R}^{\prime}\left(\boldsymbol{\mathfrak{\upsilon}}, \boldsymbol{s}, \boldsymbol{L}; \sigma \mu\right)\left(\Delta \boldsymbol{\mathfrak{\upsilon}}, \Delta \boldsymbol{s}, \Delta \boldsymbol{L}\right)=\boldsymbol{0}$, that is,
\begin{equation}\label{41}
\left[\!\! \begin{array}{cccc}
\boldsymbol{P} & \boldsymbol{0} & \boldsymbol{A}^{T} \\
\boldsymbol{A} & \boldsymbol{I} & \boldsymbol{0} \\
\boldsymbol{0} & \boldsymbol{\Lambda} & \boldsymbol{S}^{T}
\end{array}\!\!\right]\!\! \left[\!\! \begin{array}{c}
\Delta \boldsymbol{\mathfrak{\upsilon}} \\
\Delta \boldsymbol{s} \\
\Delta \boldsymbol{L}
\end{array}\!\!\right] \!\! =\! - \!\! \left[\!\! \begin{array}{c}
\boldsymbol{P} \boldsymbol{\mathfrak{\upsilon}}+\boldsymbol{G}+ \boldsymbol{A}^{T} \boldsymbol{L}\\
\boldsymbol{A} \boldsymbol{\mathfrak{\upsilon}}+\boldsymbol{s} - \boldsymbol{\theta}\\
\boldsymbol{\Lambda} \boldsymbol{S} \boldsymbol{e} -\sigma \mu \boldsymbol{e}
\end{array}\!\!\right]\!\!.
\end{equation} \par 
Subsequently, the search direction is calculated by solving 
\begin{equation}\label{50}
\left[ \begin{array}{cccc}
\!\boldsymbol{P} & \!\!\boldsymbol{0} & \!\!\boldsymbol{A}^{T}\! \\
\!\boldsymbol{A} & \!\!\boldsymbol{I} & \!\!\boldsymbol{0} \\
\!\boldsymbol{0}\! & \!\!\boldsymbol{\Lambda} & \!\!\boldsymbol{S}^{T}\!
\end{array}\!\!\right]\!\! \!\left[\!\!\! \begin{array}{c}
\!\Delta \boldsymbol{\mathfrak{\upsilon}} \\
\!\Delta \boldsymbol{s} \\
\!\Delta \boldsymbol{L}
\end{array}\!\!\!\right] \!\! =\! - \!\! \left[\!\! \begin{array}{c}
\boldsymbol{P} \boldsymbol{\mathfrak{\upsilon}}+\boldsymbol{G}+ \boldsymbol{A}^{T} \boldsymbol{L}\\
\boldsymbol{A} \boldsymbol{\mathfrak{\upsilon}}+\boldsymbol{s} - \boldsymbol{\theta}\\
\boldsymbol{\Lambda} \boldsymbol{S} \boldsymbol{e} + \Delta \boldsymbol{\Lambda}^{\mathrm {aff }} \Delta \boldsymbol{S}^{\mathrm {aff }} \boldsymbol{e} -\sigma \mu \boldsymbol{e}
\end{array}\!\!\right]\!\!.
\end{equation} \par 
By setting $\sigma = 0$ and solving \eqref{41}, $\Delta \boldsymbol{\Lambda}^{\mathrm {aff }}$ and $ \Delta \boldsymbol{S}^{\mathrm {aff }}$ can be obtained.

We can then determine $\left(\Delta \boldsymbol{\mathfrak{\upsilon}}, \Delta \boldsymbol{s}, \Delta \boldsymbol{L}\right)$, update the variables $\left(\boldsymbol{\mathfrak{\upsilon}}^{+}, \boldsymbol{s}^{+}, \boldsymbol{L}^{+}\right) = \left(\boldsymbol{\mathfrak{\upsilon}}, \boldsymbol{s}, \boldsymbol{L}\right) + \beta\left(\Delta \boldsymbol{\mathfrak{\upsilon}}, \Delta \boldsymbol{s}, \Delta \boldsymbol{L}\right)$ (where $\beta$ is chosen to ensure that $\left(\boldsymbol{s}^{+}, \boldsymbol{L}^{+}\right) > 0$), and iteratively update $\sigma$ and $\mu$ until a solution to the equations is obtained. The algorithm's instructions are outlined in \cite{Nocedal2006}. \par 
\begin{table}[htbp]
	%\caption{The nodes sorted by degree distribution}   %??
	%\label{Table.1}    %?????????
	\centering  %????
	\begin{tabular}{l}
		%??c???????????????
		\toprule[1.5pt]           %?????
		\textbf{Algorithm 1} Predictor-Corrector Algorithm for QP.  \\
		%\qquad\qquad\quad\ \ under the SCODA model \\
		\midrule%?????
		Compute $(\boldsymbol{\mathfrak{\upsilon}}(t_{\mathtt{0}}), \boldsymbol{s}(t_{\mathtt{0}}), \boldsymbol{L}(t_{\mathtt{0}}))$ with $(\boldsymbol{s}_{\mathtt{0}}, \boldsymbol{L}_{\mathtt{0}}) > 0$.\\
		\;1: \textbf{while} $t_{\mathtt{\kappa}}<t_{\mathtt{z}}$ \textbf{do}\\
		\;2: Set $(\boldsymbol{\mathfrak{\upsilon}}, \boldsymbol{s}, \boldsymbol{L}) = (\boldsymbol{\mathfrak{\upsilon}}_{\mathtt{\kappa}}, \boldsymbol{s}_{\mathtt{\kappa}}, \boldsymbol{L}_{\mathtt{\kappa}})$ and solve \eqref{41} with $\sigma = 0$ for \\ 
		\;\;\;\;\;\;\;\;\;\;\;\;\;\;\;$\left(\Delta \boldsymbol{\mathfrak{\upsilon}}^{\mathrm {aff }}, \Delta \boldsymbol{s}^{\mathrm {aff }}, \Delta \boldsymbol{L}^{\mathrm {aff }}\right)$;\\
		
		\;3: Calculate $\mu = \boldsymbol{s}^{T}\boldsymbol{L} / n$;\\
		\;4: Calculate $\hat{\beta}_{\mathrm {aff }} \!=\! \max\{\!\beta\! \in (0, 1] \;\!\! | \!\!\; (\boldsymbol{s}, \boldsymbol{L}) \!+\! \beta (\!\Delta \boldsymbol{s}^{\mathrm {aff }}, \Delta \boldsymbol{L}^{\mathrm {aff }}\!) \!\geq\! 0 \!\}$;\\
		\;5: Calculate $\mu_{\mathrm{aff}}=\left(\boldsymbol{s}+\hat{\beta}_{\mathrm{aff}} \Delta \boldsymbol{s}^{\mathrm{aff}}\right)^T\left(\boldsymbol{L}+\hat{\beta}_{\mathrm{aff}} \Delta \boldsymbol{L}^{\mathrm{aff}}\right) / n $; \\
		\;6: Set centering parameter to $\sigma = \left(\mu_{\mathrm{aff}} / \mu\right)^{3}$; \\
		\;7: Solve \eqref{50} for $\left(\Delta \boldsymbol{\mathfrak{\upsilon}}, \Delta \boldsymbol{s}, \Delta \boldsymbol{L}\right)$; \\
		\;8: Choose $\tau_{\mathtt{\kappa}} \in (0, 1)$ and set $\hat{\beta} = \min \left(\beta_{\tau_{\mathtt{\kappa}}}^{\text {pri }}, \beta_{\tau_{\mathtt{\kappa}}}^{\text {dual }}\right)$, where \\
		\;\;\;\;\;\;\;\;\;\;$\begin{aligned}
		\beta_{\tau}^{\text {pri }} & =\max \{\beta \in(0,1]: \boldsymbol{s}+\beta \Delta \boldsymbol{s} \geq(1-\tau) \boldsymbol{s}\}, \\
		\;\;\;\;\beta_{\tau}^{\text {dual }} & =\max \{\beta \in(0,1]: \boldsymbol{L}+\beta \Delta \boldsymbol{L} \geq(1-\tau) \boldsymbol{L}\};
		\end{aligned}
		$ \\
		\;9: Set $\left(\boldsymbol{\mathfrak{\upsilon}}_{\mathtt{\kappa}+1}, \boldsymbol{s}_{\mathtt{\kappa}+1}, \boldsymbol{L}_{\mathtt{\kappa}+1}\right)=\left(\boldsymbol{\mathfrak{\upsilon}}_{\mathtt{\kappa}}, \boldsymbol{s}_{\mathtt{\kappa}}, \boldsymbol{L}_{\mathtt{\kappa}}\right)+\hat{\beta}(\Delta \boldsymbol{\mathfrak{\upsilon}}, \Delta \boldsymbol{s}, \Delta \boldsymbol{L});$ \\
		10: $\mathtt{\kappa}=\mathtt{\kappa}+1$;\\
		11: \textbf{end while}\\
		\bottomrule[1.5pt]  %?????
	\end{tabular}
\end{table}

\subsection{Feasibility Analysis of Control Strategy} 
To assess the feasibility of the control strategy \eqref{43}, we expand  the expression given by \eqref{44},
\begin{equation}\label{8}
\begin{aligned}
&\frac{\mathtt{\Theta}(\boldsymbol{\zeta})+1}{(\|\boldsymbol{\zeta}+\boldsymbol{\mathfrak{d}}\|+\mathtt{r})^2}\!\! \left(\!\!(\|\boldsymbol{\zeta}+\boldsymbol{\mathfrak{d}}\|\!\!+\!\mathtt{r}) L_\mathtt{\hat{F}} \mathtt{\theta}(\boldsymbol{\zeta}) \!\! - \!\! \frac{\mathtt{\theta}(\boldsymbol{\zeta})(\boldsymbol{\zeta}+\boldsymbol{\mathfrak{d}})^T  \mathtt{\hat{F}}(\boldsymbol{\zeta})}{\|\boldsymbol{\zeta}+\boldsymbol{\mathfrak{d}}\|}\right) \\
&+ \! \frac{-(\mathtt{\Theta}(\boldsymbol{\zeta})+1) \mathtt{\theta}(\boldsymbol{\zeta})}{(\|\boldsymbol{\zeta}+\boldsymbol{\mathfrak{d}}\|+\mathtt{r})^2\|\boldsymbol{\zeta}+\boldsymbol{\mathfrak{d}}\|}(\boldsymbol{\zeta}+\boldsymbol{\mathfrak{d}})^T \mathtt{\hat{G}}(\boldsymbol{\zeta}) \boldsymbol{\mathfrak{\upsilon}}+\alpha(\mathtt{\Theta}(\boldsymbol{\zeta})) \geq 0.
\end{aligned}
\end{equation} \par

\noindent we rewrite \eqref{8}
\begin{equation}\label{9}
\begin{aligned}
&\frac{(\mathtt{\Theta}(\boldsymbol{\zeta})+1) \mathtt{\theta}(\boldsymbol{\zeta})}{(\|\boldsymbol{\zeta}+\boldsymbol{\mathfrak{d}}\|+\mathtt{r})^2\|\boldsymbol{\zeta}+\boldsymbol{\mathfrak{d}}\|}(\boldsymbol{\zeta}+\boldsymbol{\mathfrak{d}})^T \mathtt{\hat{G}}(\boldsymbol{\zeta}) \boldsymbol{\mathfrak{\upsilon}} \leq \alpha(\mathtt{\Theta}(\boldsymbol{\zeta}))+ \\
&\frac{\mathtt{\Theta}(\boldsymbol{\zeta})+1}{(\|\boldsymbol{\zeta}+\boldsymbol{\mathfrak{d}}\|+\mathtt{r})^2} \!\! \left( \!\! (\|\boldsymbol{\zeta} \! + \! \boldsymbol{\mathfrak{d}}\| \!+ \! \mathtt{r}) L_\mathtt{\hat{F}} \mathtt{\theta}(\boldsymbol{\zeta})  \!  - \! \frac{\mathtt{\theta}(\boldsymbol{\zeta})(\boldsymbol{\zeta} \! + \! \boldsymbol{\mathfrak{d}})^T \mathtt{\hat{F}}(\boldsymbol{\zeta})}{\|\boldsymbol{\zeta} \! + \! \boldsymbol{\mathfrak{d}}\|}\!\right)\!\!.
\end{aligned}
\end{equation} \par 

Based on the definition of $G_d(\boldsymbol{\zeta}):=\left(\boldsymbol{\zeta}+\boldsymbol{\mathfrak{d}}\right)^T \mathtt{\hat{G}}(\boldsymbol{\zeta})$, we proceed with a straightforward shift
\begin{equation}\label{10}
\begin{split}
G_d(\boldsymbol{\zeta}) \boldsymbol{\mathfrak{\upsilon}} \leq &\frac{(\|\boldsymbol{\zeta}+\boldsymbol{\mathfrak{d}}\|+\mathtt{r})\|\boldsymbol{\zeta}+\boldsymbol{\mathfrak{d}}\| L_\mathtt{\hat{F}} \mathtt{\theta}(\boldsymbol{\zeta})}{\mathtt{\theta}(\boldsymbol{\zeta})}-\\
&(\boldsymbol{\zeta} \! + \! \boldsymbol{\mathfrak{d}})^T \! \mathtt{\hat{F}}(\boldsymbol{\zeta}) \! + \! \frac{(\|\boldsymbol{\zeta}\! + \! \boldsymbol{\mathfrak{d}}\|+\mathtt{r})^2\|\boldsymbol{\zeta}+\boldsymbol{\mathfrak{d}}\| \alpha(\mathtt{\Theta}(\boldsymbol{\zeta}))}{(\mathtt{\Theta}(\boldsymbol{\zeta})+1) \mathtt{\theta}(\boldsymbol{\zeta})}.
\end{split}
\end{equation} \par
We then define $\left|G_{d}(\boldsymbol{\zeta})\right|:=\left(\left|G_{d1}(\boldsymbol{\zeta})\right|, \ldots, \left|G_{dq}(\boldsymbol{\zeta})\right|\right) \geq \boldsymbol{0}$. 

\begin{theorem}\label{theorem2}
	The designed control strategy could be feasible if there exists a class $\mathcal{K}$ function $\alpha(\mathtt{\Theta})$ such that 
	$$
	\begin{aligned}
	\frac{(\mathtt{\Theta}+1) \mathtt{\theta}}{(\|\boldsymbol{\zeta}\!+\!\boldsymbol{\mathfrak{d}}\| \!+ \! \mathtt{r})^2\|\boldsymbol{\zeta} \! + \! \boldsymbol{\mathfrak{d}}\|}[\left|G_{d}\right| \boldsymbol{\mathfrak{\upsilon}}_{\min } \! + \!(\boldsymbol{\zeta} \! + \! \boldsymbol{\mathfrak{d}})^T \mathtt{\hat{F}}(\boldsymbol{\zeta})-\\
	\frac{(\|\boldsymbol{\zeta}+\boldsymbol{\mathfrak{d}}\|+\mathtt{r})\|\boldsymbol{\zeta}+\boldsymbol{\mathfrak{d}}\| L_\mathtt{\hat{F}} \mathtt{\theta}}{\mathtt{\theta}}] \leq \alpha(\mathtt{\Theta}).
	\end{aligned}
	$$
\end{theorem}

\begin{proof}
	See Appendix A.
\end{proof} \par 
To our knowledge, this marks the inaugural occasion where necessary and essential conditions have been expounded to establish the feasibility of the QPs \eqref{51}, which significantly diverge from the sufficient conditions proposed in \cite{Xiao2022}. 

The necessary and sufficient conditions are presented in Theorem \ref{theorem2}. Subsequently, we define $$Y(\boldsymbol{\zeta}) \!= \!
\frac{\left[ \left|G_{d}\right| \boldsymbol{\mathfrak{\upsilon}}_{\min } \!\! + \!\!(\boldsymbol{\zeta} \!\! +\! \! \boldsymbol{\mathfrak{d}})^T \mathtt{\hat{F}} \right]\!\mathtt{\theta}\!-\!(\|\boldsymbol{\zeta}\!+\!\boldsymbol{\mathfrak{d}}\|\!+\!\mathtt{r})\|\boldsymbol{\zeta}\!+\!\boldsymbol{\mathfrak{d}}\| L_\mathtt{\hat{F}} \mathtt{\theta}}{\mathtt{\Theta}(\|\boldsymbol{\zeta}\!+\!\boldsymbol{\mathfrak{d}}\| \!+ \! \mathtt{r})^2\|\boldsymbol{\zeta} \! + \! \boldsymbol{\mathfrak{d}}\|}.
$$

\begin{theorem}\label{theorem3} 
	If $\exists N>0 $ ($N$ is a positive real number), such that $ \underset{\boldsymbol{\zeta} \in \mathtt{\bar{C}}}{\max} \{Y(\boldsymbol{\zeta})\} \leq N$, then the existence of class $\mathcal{K}$ functions $\alpha(\mathtt{\Theta})$ satisfies \eqref{11}, that is, the necessary and sufficient condition exists.
\end{theorem}
\begin{proof}
	See Appendix B.
\end{proof}

\section{Application of Proposed Control Strategy to Adaptive Cruise Control}\label{section4}
The core of this section involves implementing the proposed approach to address the challenges of the adaptive cruise control (ACC) \cite{Payman2012,Ioannou1993}. \begin{figure}[!htbp]
	\centering
	\includegraphics[width=2.8in,angle=0]{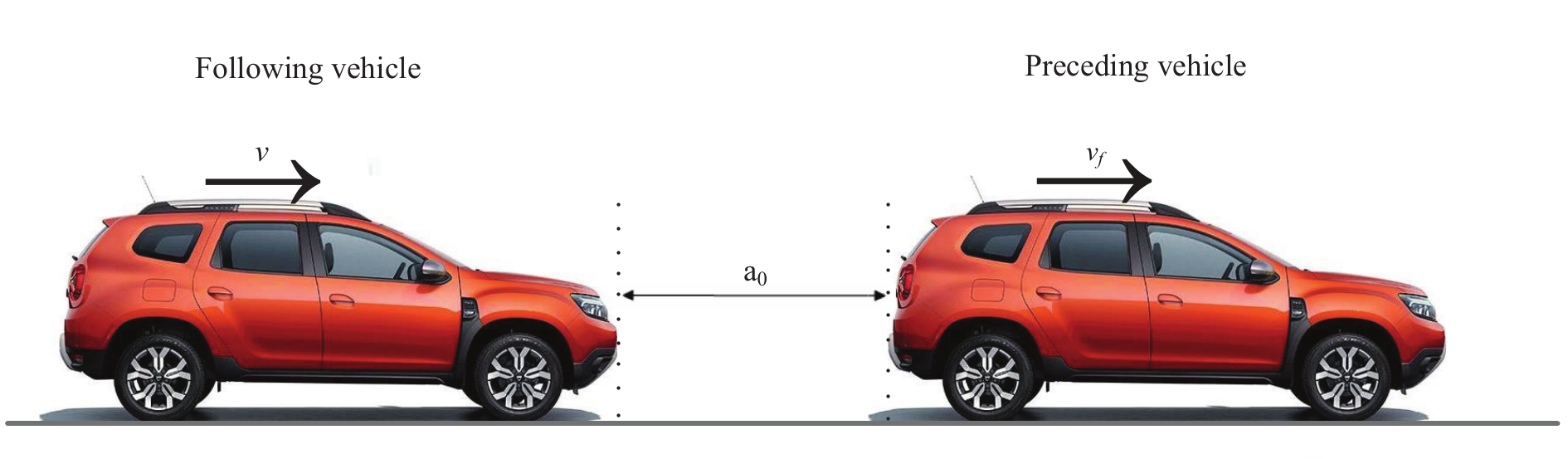}
	\captionsetup{font={small}}
	\caption{Outline of an ACC system.}
	\label{fig6}
\end{figure} \par
The value $a_0$ in Fig. \ref{fig6} represents the safe distance that vehicles should maintain to avoid collisions. The dynamics of the ACC system are described in \cite{Xiao2022}.
\begin{equation}\label{25}
\underbrace{\left[\begin{array}{l}
	\dot{\mathtt{v}}(t) \\
	\dot{\mathtt{z}}(t)
	\end{array}\right]}_{\dot{\boldsymbol{\zeta}}(t)}=\underbrace{\left[\begin{array}{c}
	-\frac{1}{\mathtt{M}} \mathfrak{F}_{\mathfrak{r}}(\mathtt{v}(t)) \\
	\mathtt{v}_{\mathtt{f}}-\mathtt{v}(t)
	\end{array}\right]}_{\mathtt{\hat{F}}(\boldsymbol{\zeta}(t))}+\underbrace{\left[\begin{array}{c}
	\frac{1}{\mathtt{M}} \\
	\mathtt{0}
	\end{array}\right]}_{\mathtt{\hat{G}}(\boldsymbol{\zeta}(t))} \mathtt{u}(t),
\end{equation} 
In the given context, ${\mathtt{M}}$ denotes the mass of the regulated automobile in kilograms (\si{kg}). Additionally, ${\mathtt{z}(t)}$ corresponds to the inter-vehicular gap in meters (\si{m}), while $\mathtt{v}_{\mathtt{f}} > 0$ and $\mathtt{v}(t) \geq0$ represent the speed of the leading car and the subsequent (controlled) vehicle, respectively, measured in meters per second (\si{m/s}). The input signal is denoted as ${\mathtt{u}(t)}$ in newtons (\si{N}). Furthermore, ${\mathfrak{F}_{\mathfrak{r}}(\mathtt{v}(t))}$, which can be estimated using the formula presented in \cite{Khalil2002}, represents the force of resistance arising from air drag or friction, expressed in newtons (\si{N}).
\begin{equation}\label{26}
{\mathfrak{F}_{\mathfrak{r}}(\mathtt{v}(t))=\mathfrak{r}_{\mathtt{0}}sgn(\mathtt{v}(t))+\mathfrak{r}_{\mathtt{1}}\mathtt{v}(t)+\mathfrak{r}_{\mathtt{2}}\mathtt{v}^{2}(t)},
\end{equation}
In this context, the magnitude of the scalars $\mathfrak{r}_{\mathtt{0}},\mathfrak{r}_{\mathtt{1}}$ and $\mathfrak{r}_{\mathtt{2}}$, which are all greater than zero, is established through empirical means. The composite form of ${\mathfrak{F}_{\mathfrak{r}}(\mathtt{v}(t))}$ signifies three distinct components: namely, the initial element representing coulombic friction, the second term characterizing viscous drag, and the third term delineating air resistance, which is proportional to the square of velocity $\mathtt{v}(t)$. \par

\textbf{Requirement} (Input and state limitations): There are constraints on the acceleration, speed, and safe distance \cite{Xiao2019,Xiao2022}.
\begin{equation}\label{27}
\begin{array}{r}
-\mathtt{c}_{\mathtt{d}} \mathtt{M} g \leq \mathtt{u}(t) \leq \mathtt{c}_{\mathtt{a}} \mathtt{M} g, \forall t \in\left[t_{\mathtt{0}}, t_{\mathtt{z}}\right], \\
\mathtt{v}_{\mathtt{min} } \leq \mathtt{v}(t) \leq \mathtt{v}_{\mathtt{max} }, \forall t \in\left[t_{\mathtt{0}}, t_{\mathtt{z}}\right], \\
\mathtt{\theta}(\boldsymbol{\zeta}(t)) = \mathtt{z}(t)- a_0 \geq 0, \forall t \in\left[t_{\mathtt{0}}, t_{\mathtt{z}}\right].
\end{array}
\end{equation}
In this given context, the variables ${\mathtt{v}_{\mathtt{min} }}$ and ${\mathtt{v}_{\mathtt{max} }}$ signify the lower and upper bounds of acceptable velocities, respectively, with both quantities expressed in non-negative values. Additionally, the coefficients ${\mathtt{c}_{\mathtt{d}}}$ and ${\mathtt{c}_{\mathtt{a}}}$, both greater than zero, correspond to the braking and acceleration rates, respectively. The symbol ${g}$ refers to the universal gravitational constant, while $a_0$ denotes a positive scalar parameter. \par 

\textbf{Objective} (Target speed): The vehicle under control tries to reach a target speed ${\mathtt{v}_\mathtt{T}}$, that is, $\left|\mathtt{v}(t)-\mathtt{v}_\mathtt{T}\right| \leq \epsilon$. \par

For the ACC system described by \eqref{25}, it is imperative to formulate a control law that accomplishes the specified \textbf{objective} while adhering to the stipulated \textbf{requirements}. \par 
Considering the dynamics provided by \eqref{25}, specific control bounds are in place (\eqref{27}), meaning that any control signal ${\mathtt{u}(t)}$ must adhere to the following condition
\begin{equation}\label{31}
-\mathtt{c}_{\mathtt{d}} g \leq \frac{\mathtt{u}(t)}{\mathtt{M}} \leq \mathtt{c}_{\mathtt{a}} g.
\end{equation} \par 
Define ${\mathtt{\theta}_1(\boldsymbol{\zeta}(t)):=\mathtt{v}_{\mathtt{max}}-\mathtt{v}(t), \mathtt{\theta}_2(\boldsymbol{\zeta}(t)):=\mathtt{v}(t)-\mathtt{v}_{\mathtt{min}}}$, and select ${\alpha(\mathtt{\theta}_1)=\mathtt{\theta}_1, \alpha(\mathtt{\theta}_2)=\mathtt{\theta}_2}$ from \eqref{5}. Then, ${\mathtt{u}(t)}$ should meet the following condition
\begin{equation}\label{29}
\begin{array}{c}
\frac{\mathfrak{F}_{\mathfrak{r}}(\mathtt{v}(t))}{\mathtt{M}} -(\mathtt{v}(t)-\mathtt{v}_{\min}) \leq \frac{\mathtt{u}(t)}{\mathtt{M}} \leq \frac{\mathfrak{F}_{\mathfrak{r}}(\mathtt{v}(t))}{\mathtt{M}}  +(\mathtt{v}_{\max}-\mathtt{v}(t)).
\end{array}
\end{equation} \par 
Let ${\mathtt{\Theta}(\boldsymbol{\zeta}(t))=e^{\frac{ \mathtt{\theta}(\boldsymbol{\zeta}(t))}{\| \boldsymbol{\zeta} + \boldsymbol{\mathfrak{d}} \| + \mathtt{r}}-\mathtt{\Delta}} -1}$, where ${\Vert \boldsymbol{\zeta} + \boldsymbol{\mathfrak{d}} \Vert=\sqrt{(\mathtt{v}(t)+d_{1})^2+(\mathtt{z}(t) + d_{2})^2}}, \mathtt{\Delta}>0$, and 
$$
\begin{aligned}
& L_\mathtt{\hat{G}} \mathtt{\Theta}(\boldsymbol{\zeta}(t)) = -\frac{(\mathtt{\Theta}(\boldsymbol{\zeta})+1) \mathtt{\theta}(\boldsymbol{\zeta})(\mathtt{v}(t) + d_{1})}{(\|\boldsymbol{\zeta}+\boldsymbol{\mathfrak{d}}\|+\mathtt{r})^2\|\boldsymbol{\zeta}+\boldsymbol{\mathfrak{d}}\|\mathtt{M}} < 0, \forall t \geq t_\mathtt{0}, \\
& G_d(\boldsymbol{\zeta}) = \left(\boldsymbol{\zeta}+\boldsymbol{\mathfrak{d}}\right)^T \mathtt{\hat{G}}(\boldsymbol{\zeta}) = \frac{\mathtt{v}(t) + d_{1}}{\mathtt{M}} > 0, \forall t \geq t_\mathtt{0}.
\end{aligned}
$$ \par  
We choose ${\alpha(\mathtt{\Theta})=K\mathtt{\Theta}}$ according to \eqref{10}, where $K$ is a positive constant. Substituting this into \eqref{44}, we have
\begin{equation}\label{30}
L_{\mathtt{\hat{F}}} \mathtt{\Theta}(\boldsymbol{\zeta}(t))+L_{\mathtt{\hat{G}}} \mathtt{\Theta}(\boldsymbol{\zeta}(t)) \mathtt{u}(t)+K \mathtt{\Theta}(\boldsymbol{\zeta}(t)) \geq 0.
\end{equation}

Considering a Lyapunov function ${\mathtt{V}_{\!\mathtt{ACC}}(\boldsymbol{\zeta}(t))\!=\!(\mathtt{v}(t)\!-\!\mathtt{v}_\mathtt{T})^2}$, where ${\mathtt{\chi}_1=\mathtt{\chi}_2=1}$ and ${\mathtt{\chi}_3>0}$ in Definition. \ref{definition4} and \eqref{46}, we obtain
\begin{equation}\label{28}
\begin{array}{l}
L_{\mathtt{\hat{F}}} \mathtt{V}_{\mathtt{ACC}}(\boldsymbol{\zeta}(t))+L_{\mathtt{\hat{G}}} \mathtt{V}_{\mathtt{ACC}}(\boldsymbol{\zeta}(t)) \mathtt{u}(t)+ \\ \mathtt{\chi}_3 \mathtt{V}_{\mathtt{ACC}}(\boldsymbol{\zeta}(t)) 
\leq \delta_{\mathtt{ACC}}(\mathtt{v}), \forall t \in[t_\mathtt{0},t_\mathtt{z}],
\end{array}
\end{equation} 
where ${\delta_{\mathtt{ACC}}(\mathtt{v})}$ is a slack variable, rendering \eqref{28} into a permissive CLF. \par 

Accordingly, we define the optimal control strategy for the ACC system through the following optimization problem
\begin{equation}\label{33}
\begin{array}{c}
\underset{\mathtt{u}(t), \delta_{\mathtt{ACC}}(\mathtt{v})}{\min} \int_{t_{\mathtt{0}}}^{t_{\mathtt{z}}}\left(\frac{\mathtt{u}(t)-\mathfrak{F}_{\mathfrak{r}}(\mathtt{v}(t))}{\mathtt{M}}\right)^{2}+\mathtt{p} \delta^{2}_{\mathtt{ACC}}(\mathtt{v}) d t, \\
\text { s.t. }\eqref{31},\eqref{29},\eqref{30},\eqref{28}.
\end{array}
\end{equation} \par 
We discretize time, assuming that the velocity $\mathtt{v}(t)$ remains constant within each adjacent small time interval $\left[t_{\mathtt{\kappa}}, t_{\mathtt{\kappa}+1}\right), \mathtt{\kappa}=0,1,2, \ldots, t_{\mathtt{0}}=0$. Since the selection of the optimal control law $\boldsymbol{\mathfrak{\upsilon}}^{*}(t)$ corresponds to each individual time interval, the term $\frac{\mathfrak{F}^{2}_{\mathfrak{r}}(\mathtt{v}(t))}{\mathtt{M}^{2}}$ within the integrand of \eqref{33} is a constant and need not be considered. Consequently, \eqref{33} may be simplified in the $\left[t_{\mathtt{\kappa}}, t_{\mathtt{\kappa}+1}\right)$ interval as presented in literature
\begin{equation}\label{35}
\begin{array}{l}
\left(\frac{\mathtt{u}(t)-\mathfrak{F}_{\mathfrak{r}}(\mathtt{v}(t))}{\mathtt{M}}\right)^{2}+\mathtt{p} \delta^{2}_{\mathtt{ACC}}(\mathtt{v}) \rightarrow \\
{\left[\begin{array}{ll}
	\!\!\! \mathtt{u}(t) & \delta_{\mathtt{ACC}}(\mathtt{v})
	\!\!\! \end{array}\right] \!\! \left[\begin{array}{cc}
	\!\!\! \frac{1}{\mathtt{M}^{2}} & \!\!\! 0 \\
	\!\!\! 0 & \!\!\!\!\! \mathtt{p}
	\!\!\! \end{array} \!\! \right]\!\!\! \left[\begin{array}{l}
	\mathtt{u}(t) \\
	\!\!\! \delta_{\mathtt{ACC}}(\mathtt{v})
	\end{array} \!\!\! \right] \!\! + \!\! \left[\begin{array}{cc}
	\!\!\! \frac{-2 \mathfrak{F}_{\mathfrak{r}}(\mathtt{v}(t))}{\mathtt{M}^{2}} & 0
	\end{array} \!\!\! \right] \!\!\! \left[\begin{array}{c}
	\!\! \mathtt{u}(t) \\
	\!\! \delta_{\mathtt{ACC}}(\mathtt{v})
	\end{array} \!\!\! \right]}.
\end{array}
\end{equation} \par 
Therefore, the model of the optimization problem \eqref{33} for each time interval can be rewritten as
\begin{equation}\label{36}
\begin{gathered}
\boldsymbol{\mathfrak{\upsilon}}^{*}(t)=\arg \min _{\boldsymbol{\mathfrak{\upsilon}}(t)} \frac{1}{2} \boldsymbol{\mathfrak{\upsilon}}(t)^{T} \boldsymbol{P} \boldsymbol{\mathfrak{\upsilon}}(t)+\boldsymbol{G}^{T} \boldsymbol{\mathfrak{\upsilon}}(t) \\
\text { s.t. } \boldsymbol{A}_{i}^{T}  \boldsymbol{\mathfrak{\upsilon}}(t) \leq \mathtt{\theta}_{i},
\end{gathered}
\end{equation}
\noindent where
\begin{align}
\boldsymbol{\mathfrak{\upsilon}}(t)\!\!=\!\! \left[ \! \! \begin{array}{c}
\!\mathtt{u}(t) \\ \notag
\!\delta_{\mathtt{ACC}}(\mathtt{v})
\end{array}\! \!\! \right]\!, \!\boldsymbol{P} \!=\!\! \left[\! \! \begin{array}{cc}
\!\frac{2}{\mathtt{M}^{2}} & \!0 \\ \notag
\!0 & \!2 \mathtt{p}
\end{array} \! \!\! \right]\!, \!\boldsymbol{G}\!=\!\! \left[\! \begin{array}{c}
\!\frac{-2 \mathfrak{F}_{\mathfrak{r}}(\mathtt{v}(t))}{\mathtt{M}^{2}} \\ \notag
\!0
\end{array}\!\! \right],
\end{align}

\noindent and $\boldsymbol{A}_{i}, \mathtt{\theta}_{i}$ are the values acquired by converting inequalities \eqref{31}, \eqref{29}, \eqref{30}, and \eqref{28} into $\boldsymbol{A}_{i}\boldsymbol{\upsilon}\leq \mathtt{\theta}_{i}$. Fig. \ref{fig11} illustrates the general procedure of employing the proposed method for constraint-based optimal control.

\begin{figure}[!htbp]
	\centering
	\includegraphics[width=2.8in,angle=0]{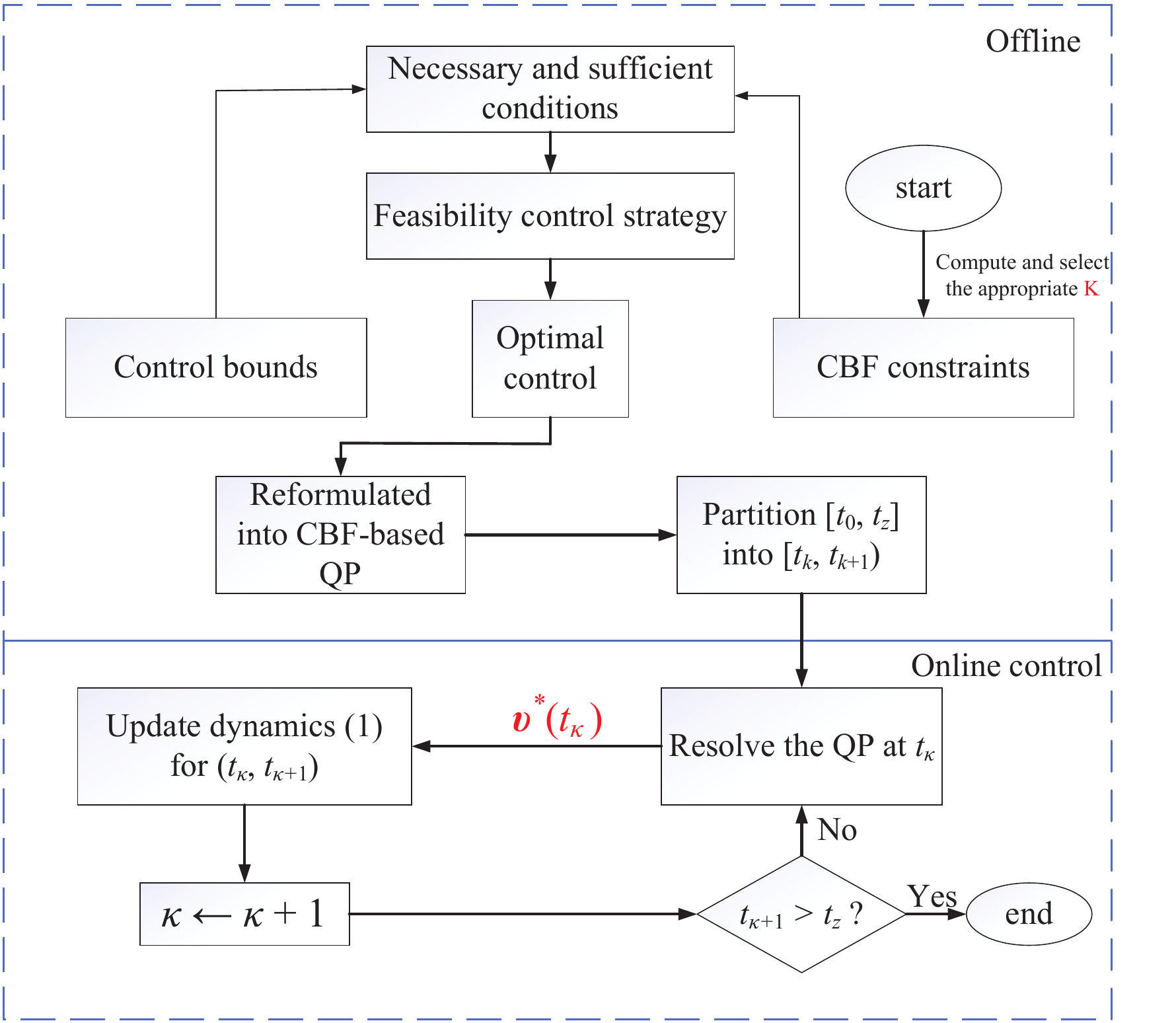}
	\captionsetup{font={small}}
	\caption{The overall process of using the feasibility control strategy.}
	\label{fig11}
\end{figure} \par

\section{Numerical Simulation}\label{section5}
In order to enhance the applicability to real-world scenarios, multiple initial positions were considered in the simulations of an adaptive cruise control system. The quadratic program problem \eqref{36} was computed subject to the constraint equations
$
A_{\mathrm{clf}} \boldsymbol{\mathfrak{\upsilon}}(t) \leq \theta_{\mathrm{clf}},
A_{\text {ncbf }} \boldsymbol{\mathfrak{\upsilon}}(t) \leq \mathtt{\theta}_{\text {ncbf }},
A_{\mathrm{v}_{\max }} \boldsymbol{\mathfrak{\upsilon}}(t) \leq \mathtt{\theta}_{\mathrm{v}_{\max }},
A_{\mathrm{v}_{\min }} \boldsymbol{\mathfrak{\upsilon}}(t) \leq \mathtt{\theta}_{\mathrm{v}_{\min }},
A_{\mathrm{limit}} \boldsymbol{\mathfrak{\upsilon}}(t) \leq \mathtt{\theta}_{\mathrm{limit}}.
$ \par 

We supposed that the speed of the preceding vehicle was ${\mathtt{v}_\mathtt{f}=13.89\si{m/s}}$ (\cite{Xiao2022,TaylorAmes2020}) and the rear car had different initial speeds $\mathtt{v}(0)$. The initial distance $\mathtt{z}(0)$ was set as $100 \si{m}$, and the safe distance was $a_0 = 10 \si{m}$. The remaining parameters were as follows: $\mathtt{v}_{\mathtt{max}} = 55\si{m/s}$, $\mathtt{v}_{\mathtt{min}} = 0\si{m/s}, \mathtt{v}_\mathtt{T} = 24\si{m/s}$, $g = 9.81\si{m/s^{2}}, \mathtt{M}=1650\si{kg}, \mathfrak{r}_{\mathtt{0}}=0.1\si{N}, \mathfrak{r}_{\mathtt{1}}=5\si{Ns/m}, \mathfrak{r}_{\mathtt{2}}=0.25\si{Ns^2/m}, \mathtt{\Delta} t=0.1\si{s}, \mathtt{p}=1$, $\mathtt{c}_\mathtt{d} = \mathtt{c}_\mathtt{a} = 0.4, \mathtt{\chi}_{1}=1, \mathtt{\chi}_{2}=1, \mathtt{\chi}_{3} =10, \mathtt{\Delta}=0.09, K=0.2, \mathtt{r} = 0.01, \boldsymbol{\mathfrak{d}} = \left(0.1, 0.1\right)$. Figs. \ref{fig4}-\ref{fig7} display the simulation results. \par  
\begin{figure}[!htbp]
\centering
\subfigure[]{
	\label{fig4:a}
	\includegraphics[width=2.8in,angle=0]{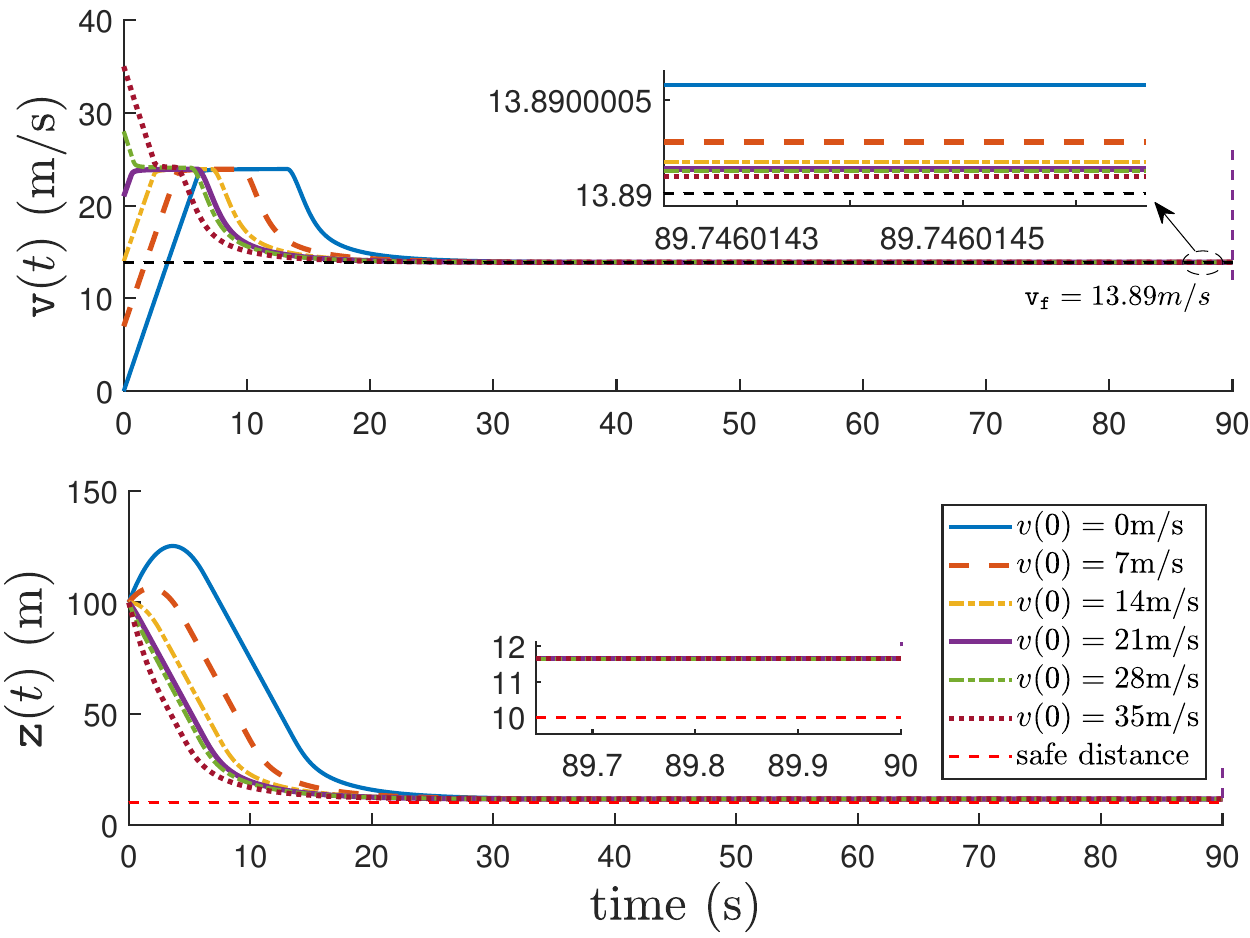}}
\subfigure[]{
	\label{fig4:b}
	\includegraphics[width=2.8in,angle=0]{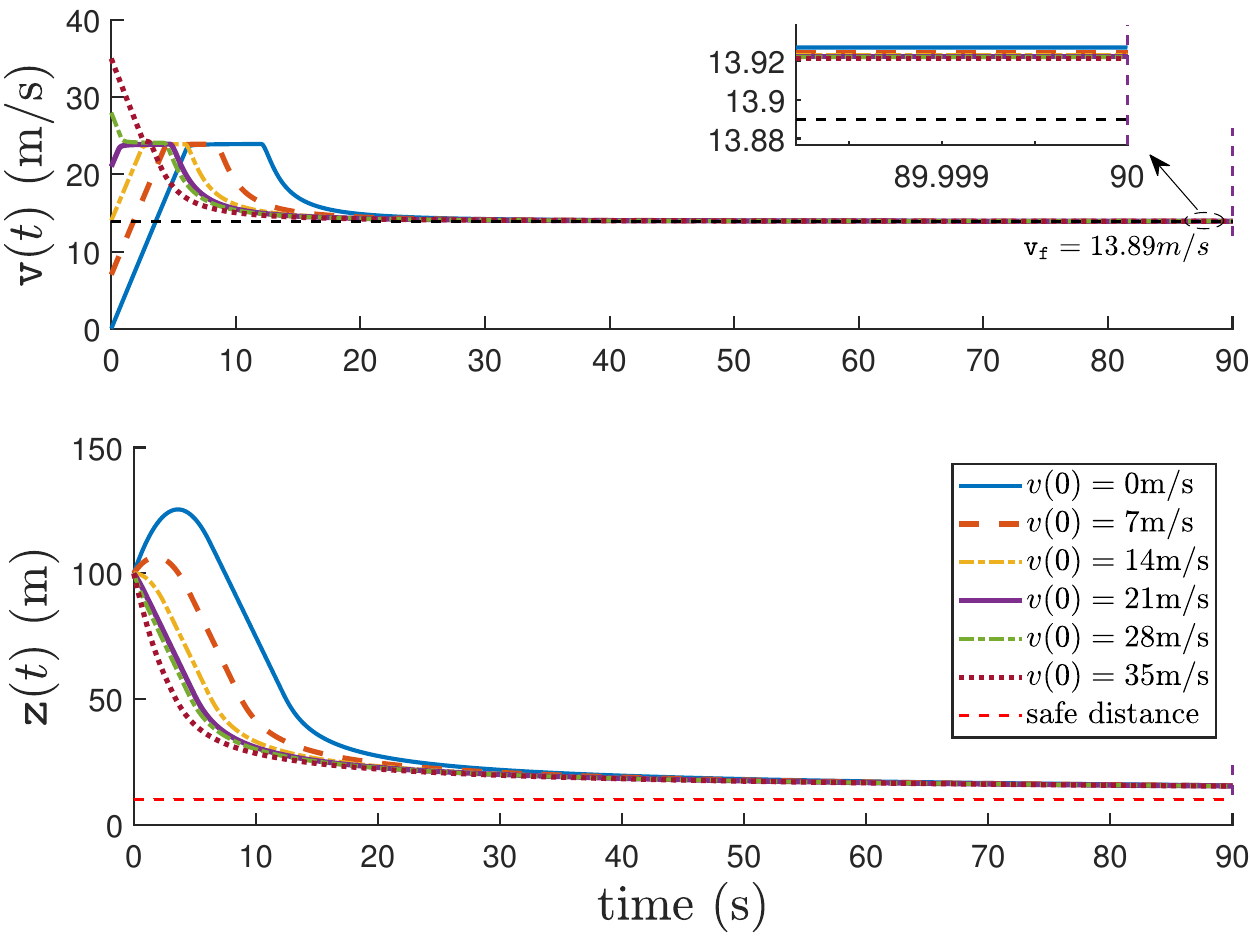}}
\captionsetup{font={small}}
\caption{ (a) NCBF: The ACC system states for different velocities. (b) HOCBF: The ACC system states for different velocities.}
\label{fig4}
\end{figure} \par

\begin{figure}[!htbp]
	\centering
	\subfigure[]{
		\label{fig1:a}
		\includegraphics[width=2.8in,angle=0]{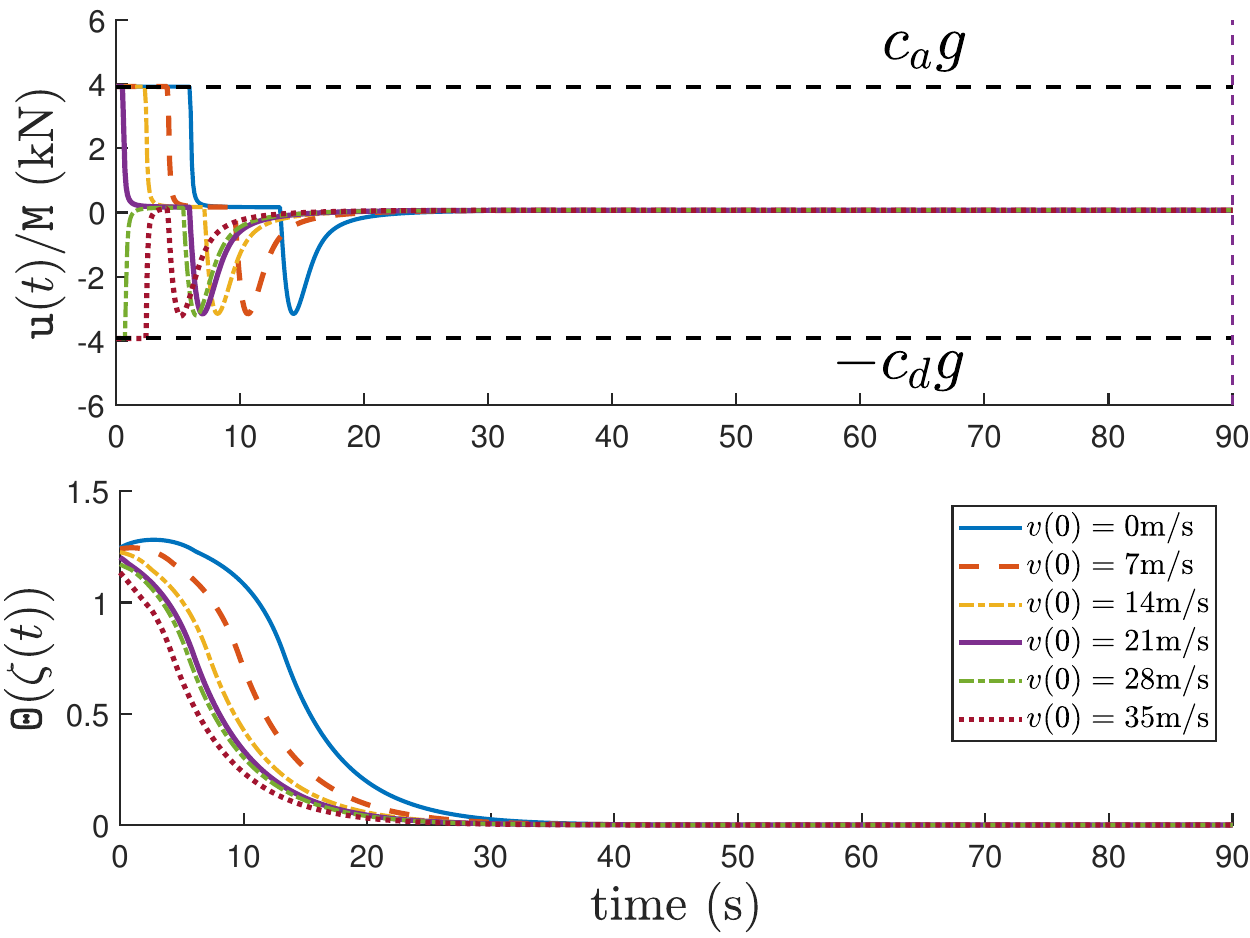}}
	\subfigure[]{
		\label{fig1:b}
		\includegraphics[width=2.8in,angle=0]{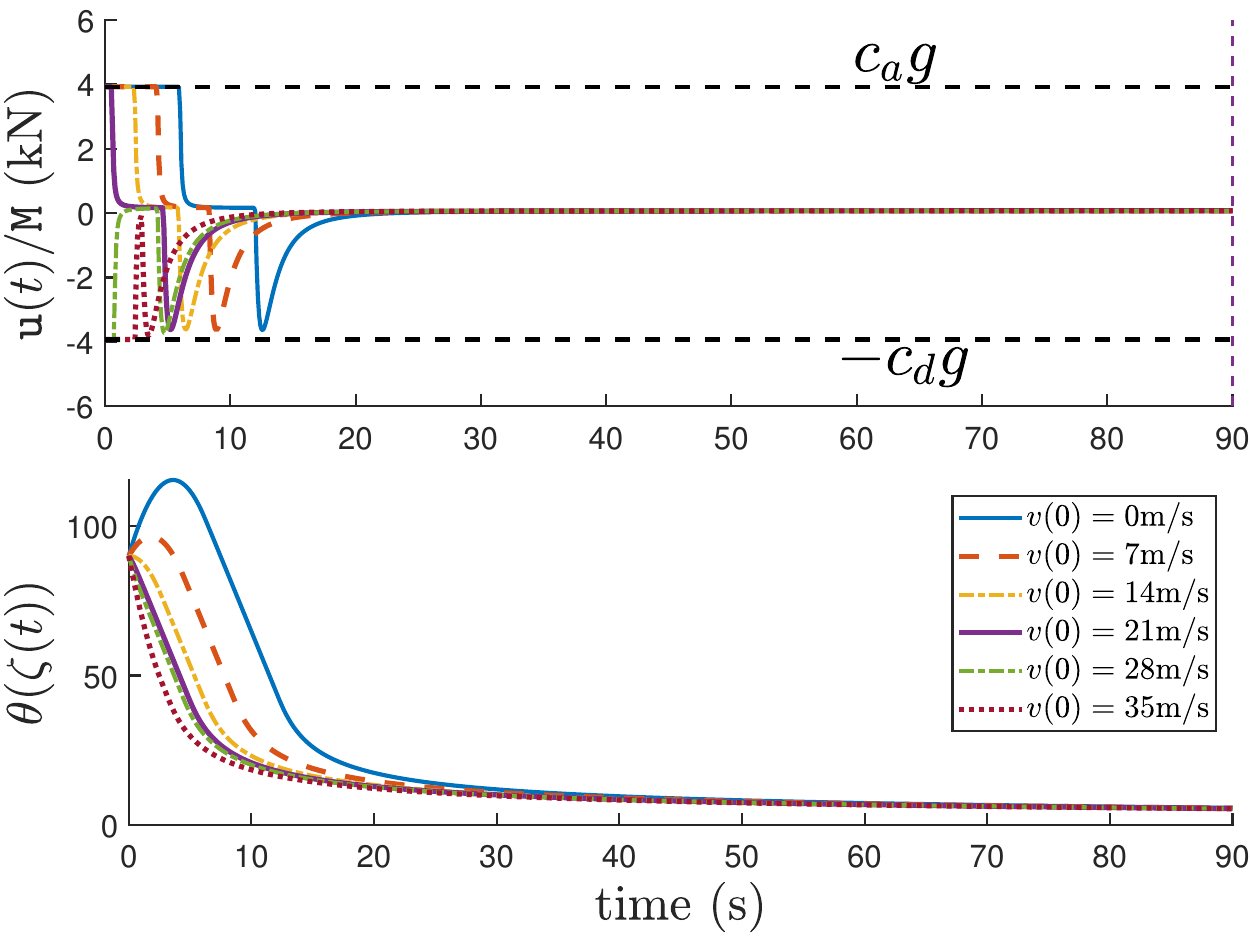}}
	\captionsetup{font={small}}
	\caption{ (a) NCBF: The control input and barrier function for different initial velocities. (b) HOCBF: The control input and barrier function for different initial velocities.}
	\label{fig1}
\end{figure} \par

\begin{figure}[!htbp]
	\centering
	\subfigure[]{
		\label{fig2:a}
		\includegraphics[width=3in,angle=0]{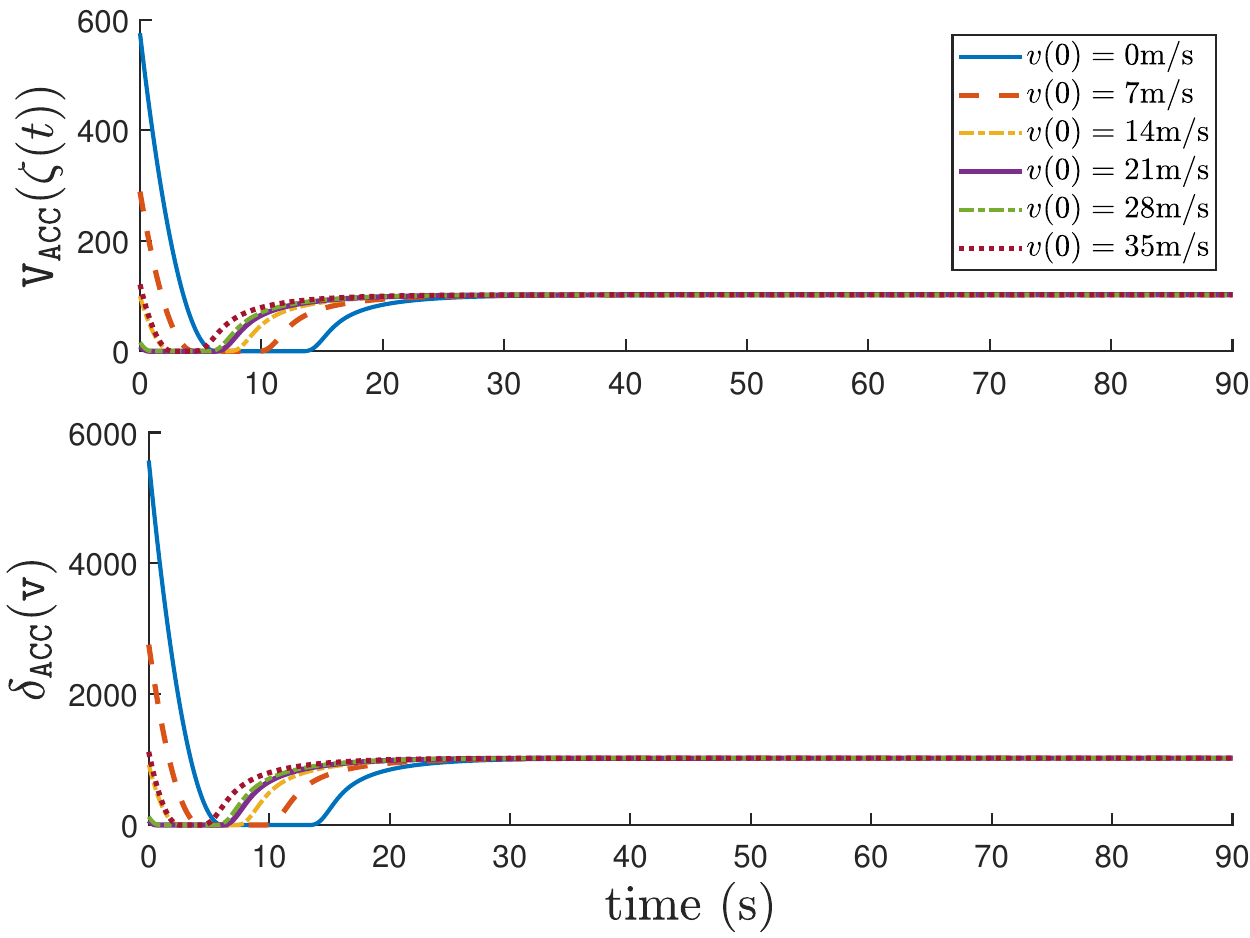}}
	\subfigure[]{
		\label{fig2:b}
		\includegraphics[width=3in,angle=0]{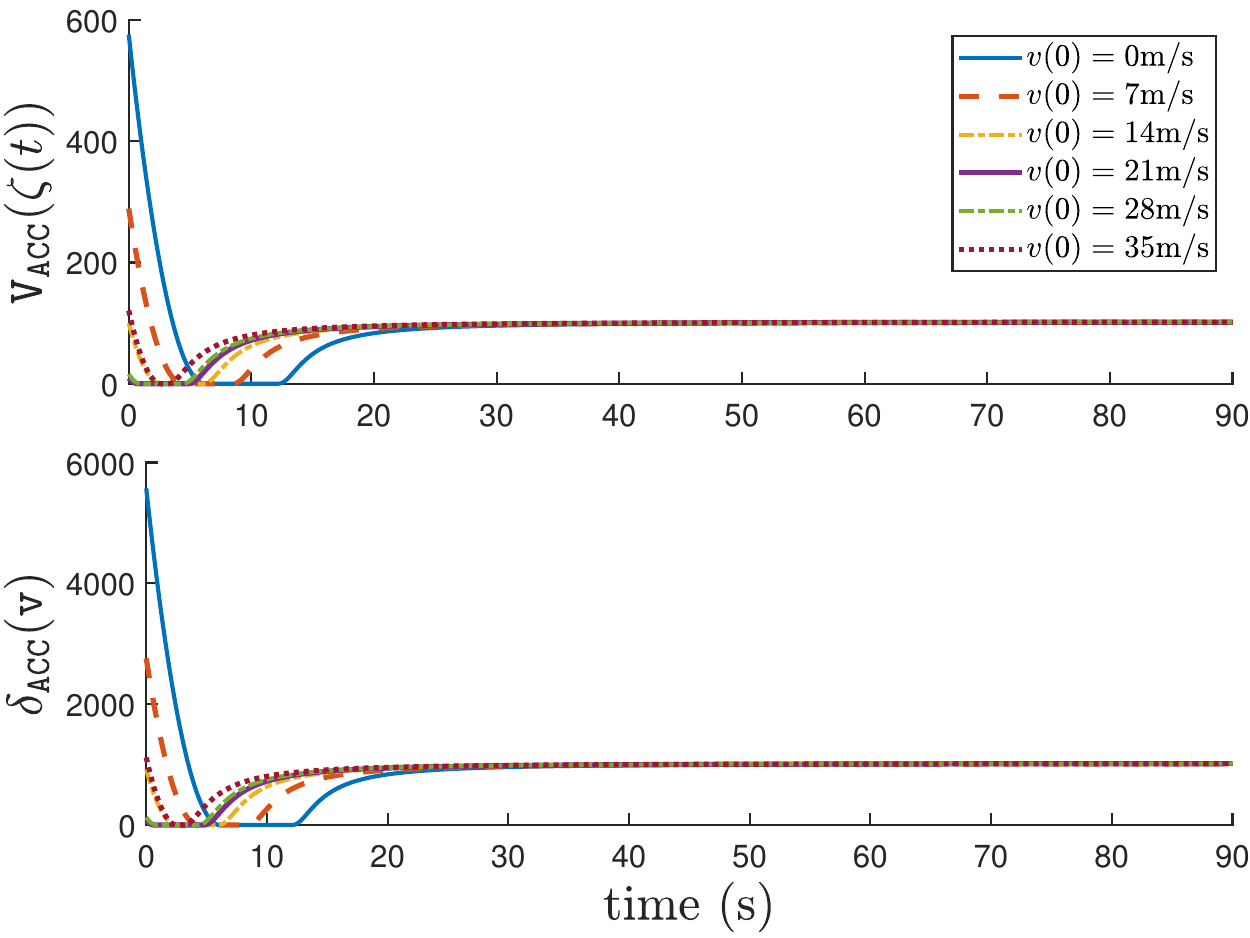}}
	\captionsetup{font={small}}
	\caption{ (a) NCBF: $\mathtt{V}_{\mathtt{ACC}}(\boldsymbol{x}(t)), \delta_{\mathtt{A C C}}(\mathtt{v}(t))$ for different initial velocities. (b) HOCBF: $\mathtt{V}_{\mathtt{ACC}}(\boldsymbol{x}(t)), \delta_{\mathtt{A C C}}(\mathtt{v}(t))$ for different initial velocities.}
	\label{fig2}
\end{figure} \par

Fig. \ref{fig4} illustrates the accelerated convergence of the ACC system to a safe and stable state achieved by the designed NCBF, resulting in faster convergence. Regardless of the chosen control barrier functions or initial velocity, as demonstrated in Figs. \ref{fig4:a} and \ref{fig4:b}, the system consistently prioritizes reaching the target velocity. However, since the expected velocity alone does not guarantee safety, the system applies the optimal solution $\boldsymbol{\upsilon}^{*}(t)$ obtained by solving the quadratic programming problem with inequality constraints to the controlled vehicle. This aligns its velocity with the preceding vehicle, ensuring safety and stability in the control process.  \par 

In Fig. \ref{fig1:a} and \ref{fig1:b}, both the NCBF and HOCBF methods ensure the feasibility of quadratic programming. Furthermore, as observed in Fig. \ref{fig1} and \ref{fig2}, starting from a certain point, the value of $\mathtt{\Theta}(\boldsymbol{\zeta}(t))$ continuously decreases, eventually approaching zero as its derivative tends to zero. By appropriately selecting the parameter $\mathtt{p}_{\mathtt{ACC}}$, the inequality-constrained quadratic programming, which includes constraints derived from the control Lyapunov function and control barrier function, prioritizes minimizing the term $\mathtt{p}_{\mathtt{ACC}} \delta^{2}_{\mathtt{ACC}}(\mathtt{v})$. In other words, it ensures that the velocity, $\mathtt{v}(t)$, reaches the desired state first. After that point, $\mathtt{\Theta}(\boldsymbol{\zeta}(t))$ continues to decrease monotonically. Additionally, Figs. \ref{fig1} and \ref{fig2} demonstrate the effectiveness of the CLF-CBF-QP method in addressing multitarget tasks. \par 

After analyzing Figure \ref{fig3}, it becomes clear that both NCBF and HOCBF guarantee the system's safety. However, the optimal solution achieved with NCBF demonstrates a relatively smoother rate of change. Consequently, the distance between the two vehicles is smaller once the ACC system reaches a safe and stable state compared to the distance obtained with HOCBF. In set $\mathtt{\bar{C}}$, this is expressed as follows: NCBF brings the system's state closer to the safety boundary of the set $\partial \mathtt{\bar{C}} =\{\boldsymbol{\zeta} \in \mathbb{R}^{\mathtt{n}}: \theta(\boldsymbol{\zeta})=0 \}$. While HOCBF also ensures safety, it results in a more cautious safety distance, representing a more conservative set. This suggests that NCBF can be considered the 'optimal' CBF, whereas HOCBF is relatively more conservative \cite{AmesXu2017}. \par 

Fig. \ref{fig7} clearly demonstrates the effectiveness of the designed inequality-constrained optimization control \eqref{33} and quadratic programming \eqref{36} in meeting the requirements. Through appropriate parameter tuning, the CLF-CBF-QP approach ensures the safe and stable control of the ACC system, as depicted in the figure. The control Lyapunov function guides the controlled vehicle's speed, $\mathtt{v}(t)$, to reach the desired velocity initially, resulting in a continuous decrease in the distance,  $\mathtt{z}(t)$, between the two vehicles. As safety is strictly defined as a hard constraint, the control barrier function compels the speed of the controlled vehicle to synchronize with that of the preceding vehicle, thereby ensuring the overall safety of the system.

\begin{figure}[!htbp]
	\centering
	\subfigure[]{
		\label{fig3:a}
		\includegraphics[width=3in,angle=0]{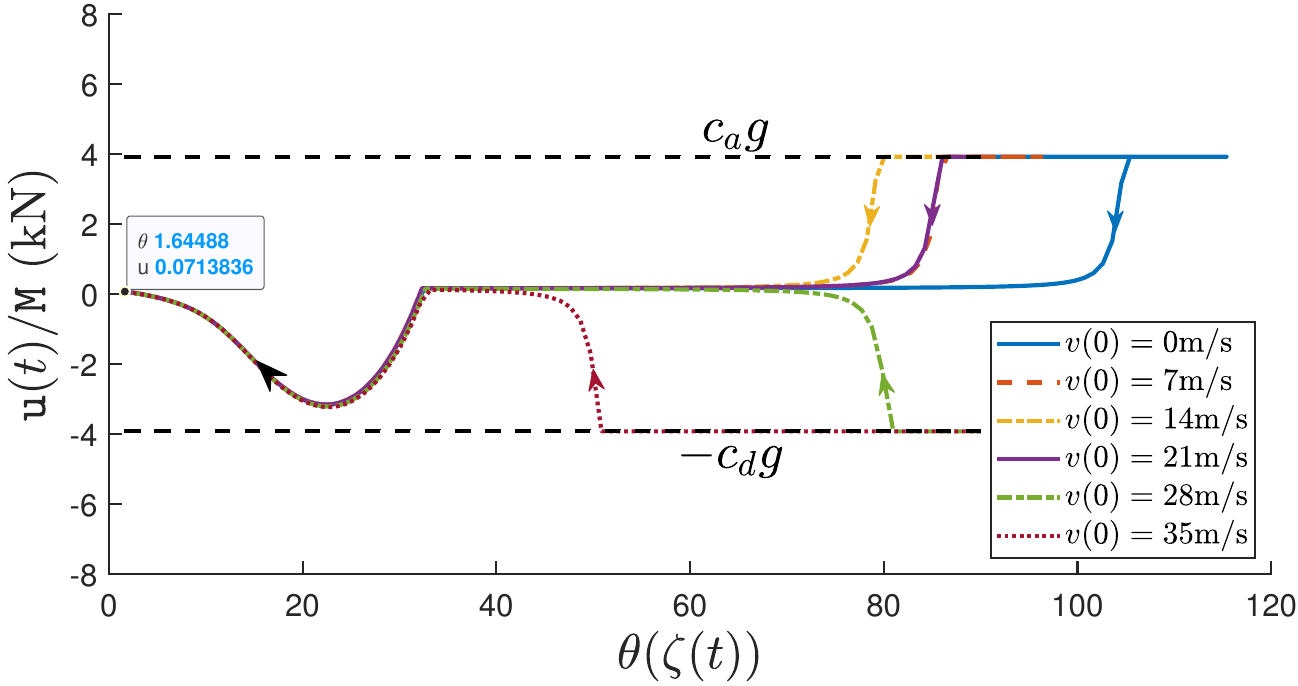}}
	\subfigure[]{
		\label{fig3:b}
		\includegraphics[width=3in,angle=0]{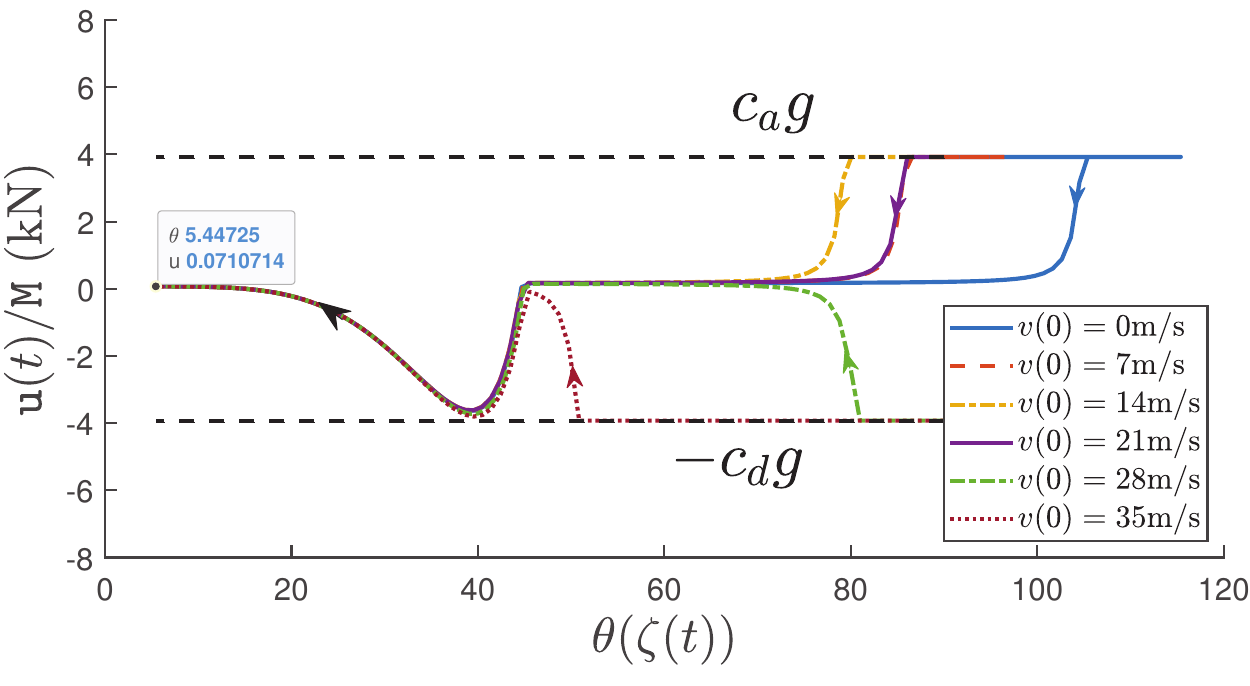}}
	\captionsetup{font={small}}
	\caption{ (a) NCBF: The temporal evolution of the relationship between $\mathtt{u}(t)$ and $\mathtt{\theta}(\boldsymbol{\zeta}(t))$ for different initial velocities is depicted in the evolutionary diagram. The arrows illustrate the directional trend of $\mathtt{\theta}(\boldsymbol{\zeta}(t))$ over time. (b) HOCBF: The temporal evolution of the relationship between $\mathtt{u}(t)$ and $\mathtt{\theta}(\boldsymbol{\zeta}(t))$ for different initial velocities is depicted in the evolutionary diagram. The arrows illustrate the directional trend of $\mathtt{\theta}(\boldsymbol{\zeta}(t))$ over time.}
	\label{fig3}
\end{figure} \par

\begin{figure}[!htbp]
	\centering
	\subfigure[]{
		\label{fig7:a}
		\includegraphics[width=3.2in,angle=0]{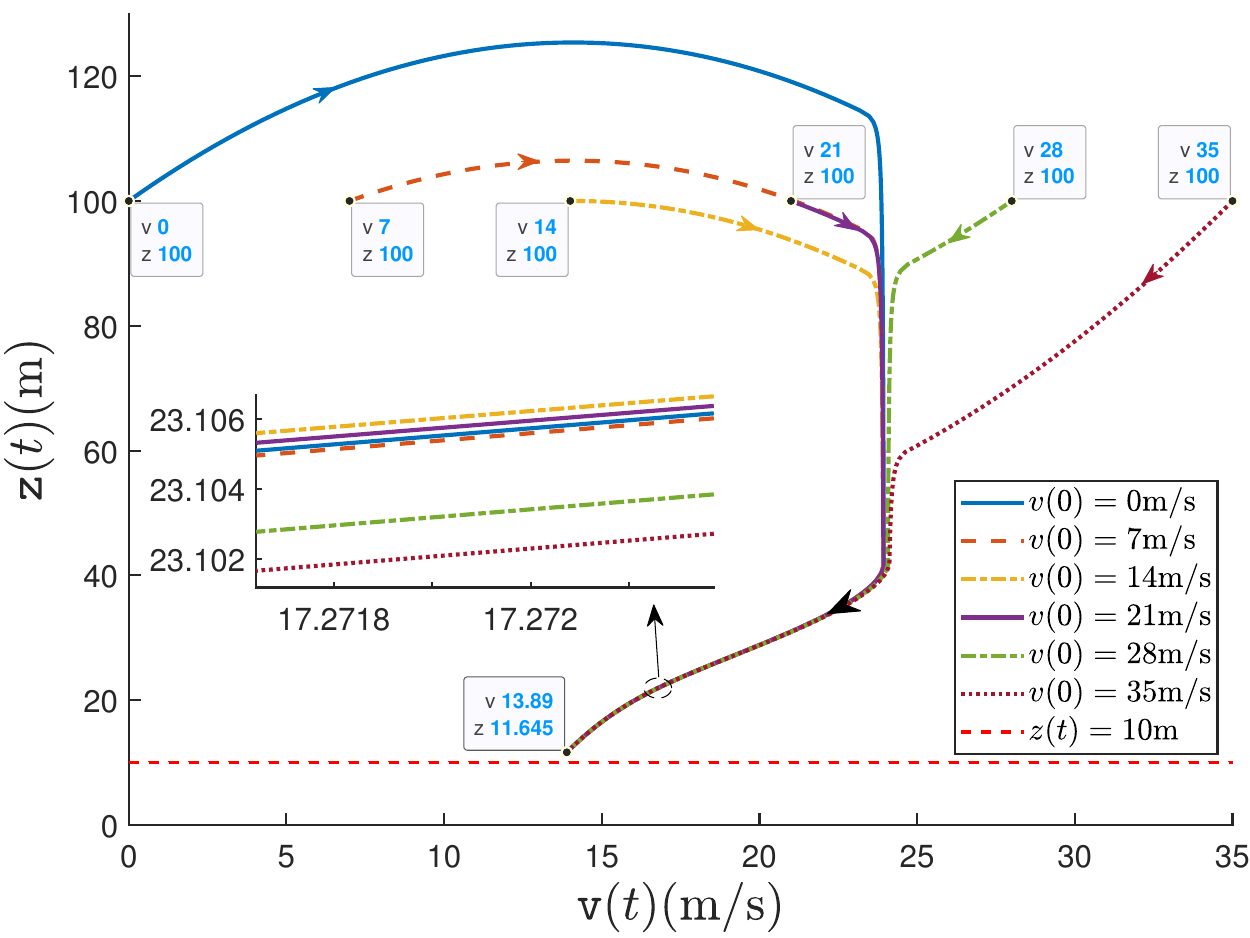}}
	\subfigure[]{
		\label{fig7:b}
		\includegraphics[width=3.2in,angle=0]{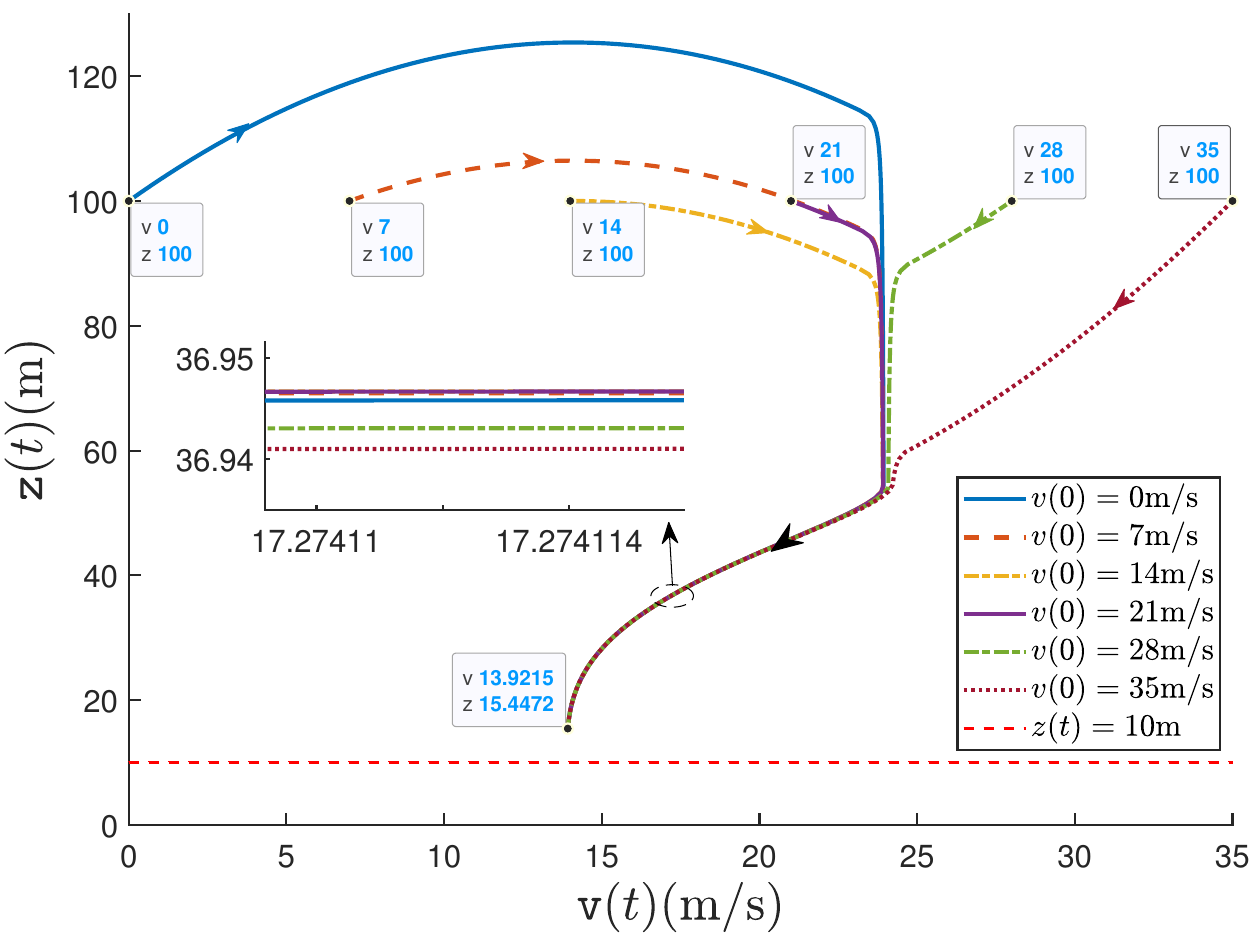}}
	\captionsetup{font={small}}
	\caption{ (a) NCBF: The correlation between $\mathtt{z}(t)$ and $\mathtt{v}(t)$ is represented, where the arrows indicate the directional trend of $\mathtt{z}(t)$ over time. (b) HOCBF: The correlation between $\mathtt{z}(t)$ and $\mathtt{v}(t)$ is represented, where the arrows indicate the directional trend of $\mathtt{z}(t)$ over time.}
	\label{fig7}
\end{figure} \par

\section{Conclusion}\label{section6}
In this study, we have introduced a paradigmatic nonlinear control barrier function (NCBF) framework for affine nonlinear systems. By operating under the notable feature of possessing a relative degree of 1, the NCBF constraint works towards ensuring robust system safety while adhering to this requirement. The presented approach delineates a distinctive departure from extant techniques in that it is distinctly applicable to affine nonlinear systems and provides both indispensable and adequate conditions essential for the meticulous design of feasible control protocols. In demonstrating its veracity, we deployed the NCBF on an adaptive cruise control system, whereby it proffered quantitative prove of stability. Nonetheless, it warrants mentioning that the solution of the constrained quadratic program problem may entail considerable computational overhead. Future inquiries will address this predicament by formulating an efficient and computationally practical control strategy capable of addressing both control input constraints and the NCBF constraint. A reduction in computational burden will augur well with the pragmatic applicability of the control strategy whilst ensuring optimal system longevity and stability.

\section*{Appendix}

\subsection{}\label{A}

\textit{Proof of \textbf{Theorem \ref{theorem2}}:} 
	The viability of the control methodology \eqref{43} hinges on the modification of any symbol within the temporal discontinuity, within the construct of $G_d(\boldsymbol{\zeta})$. \par 
	\textbf{(1) No element in $G_d(\boldsymbol{\zeta})$ changes its symbol: }
	If $G_d(\boldsymbol{\zeta})\geq \boldsymbol{0}$ ($G_d(\boldsymbol{\zeta}) = \left|G_{d}(\boldsymbol{\zeta})\right|$), we multiply $G_d(\boldsymbol{\zeta})$ on the left side of \eqref{45}, we obtain
	\begin{equation}\label{13}
	G_d(\boldsymbol{\zeta}) \boldsymbol{\mathfrak{\upsilon}}_{\min } \leq G_d(\boldsymbol{\zeta}) \boldsymbol{\mathfrak{\upsilon}} \leq G_d(\boldsymbol{\zeta})
	\boldsymbol{\mathfrak{\upsilon}}_{\max }.
	\end{equation} \par 	
	The sufficient and necessary condition under which \eqref{10} and \eqref{13} do not conflict is
	\begin{equation}\label{15}
	\begin{aligned}
	\frac{(\mathtt{\Theta}+1) \mathtt{\theta}}{(\|\boldsymbol{\zeta}\!+\!\boldsymbol{\mathfrak{d}}\| \!+ \! \mathtt{r})^2\|\boldsymbol{\zeta} \! + \! \boldsymbol{\mathfrak{d}}\|}[G_d \boldsymbol{\mathfrak{\upsilon}}_{\min } \! + \!(\boldsymbol{\zeta} \! + \! \boldsymbol{\mathfrak{d}})^T \mathtt{\hat{F}}(\boldsymbol{\zeta})-\\
	\frac{(\|\boldsymbol{\zeta}+\boldsymbol{\mathfrak{d}}\|+\mathtt{r})\|\boldsymbol{\zeta}+\boldsymbol{\mathfrak{d}}\| L_\mathtt{\hat{F}} \mathtt{\theta}}{\mathtt{\theta}}] \leq \alpha(\mathtt{\Theta}).
	\end{aligned}
	\end{equation}
	Here, $\mathtt{\Theta}(\boldsymbol{\zeta}), \mathtt{\theta}(\boldsymbol{\zeta})$, and $G_d(\boldsymbol{\zeta})$ are abbreviated as $\mathtt{\Theta}$, $\mathtt{\theta}$, and $G_d$, respectively. \par 
	
	\begin{figure}[!htbp]
		\centering
		\includegraphics[width=3in,angle=0]{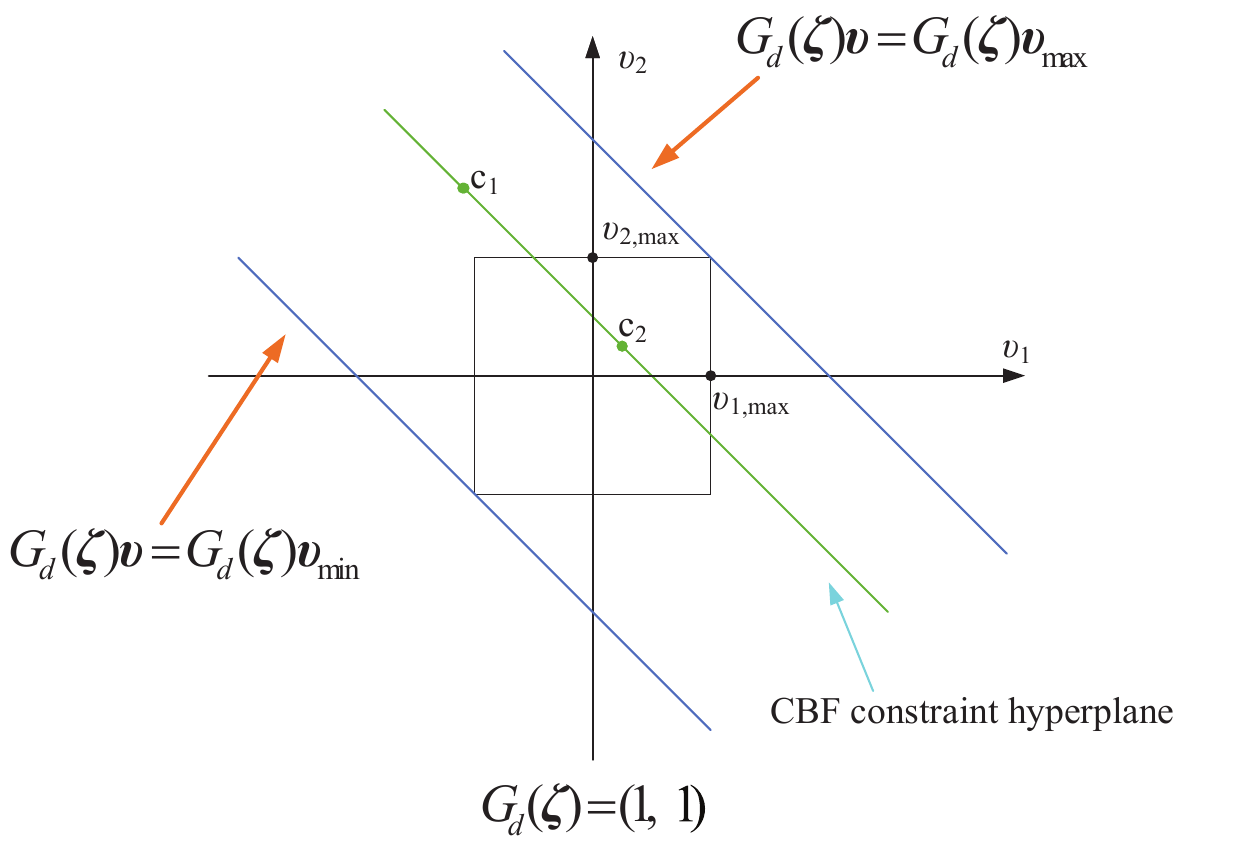}
		\captionsetup{font={small}}
		\caption{The correlation between the confinement of control input \eqref{45} and the constraint of Control Barrier Function (CBF) \eqref{44} in a scenario where a two-dimensional vector $\boldsymbol{\mathfrak{\upsilon}}=(\mathtt{\upsilon}_1, \mathtt{\upsilon}_2)^{T}$ is employed, under the condition that $G_{d}(\boldsymbol{\zeta}) \geq \boldsymbol{0}$. The inclination of the two blue lines or hyperplanes $G_d(\boldsymbol{\zeta}) \boldsymbol{\mathfrak{\upsilon}} = G_d(\boldsymbol{\zeta}) \boldsymbol{\mathfrak{\upsilon}}_{\min}$ and $G_d(\boldsymbol{\zeta}) \boldsymbol{\mathfrak{\upsilon}} = G_d(\boldsymbol{\zeta})\boldsymbol{\mathfrak{\upsilon}}_{\max}$ hinges on the magnitude of $G_{d}(\boldsymbol{\zeta})$. Meanwhile, the domain of the control input limitation \eqref{45} translates to the geometrically enclosed area of a rectangle inferred by the aforementioned hyperplanes.}
		\label{fig0}
	\end{figure} \par
	Based on the observations made in Fig. \ref{fig0}, we can make the inference that if \eqref{10} does not contradict with \eqref{13}, then it effectively becomes conflict-free from \eqref{45}, as represented by $c_1$ and $c_2$ in Fig. \ref{fig0}. Consequently, through this argument, we can discern that \eqref{15} is both an essential and adequate condition for establishing the compatibility between \eqref{45} and \eqref{44}.
	
	\textbf{(2) Symbol changes occur in some components of $G_d(\boldsymbol{\zeta})$: }Remember that $G_d(\boldsymbol{\zeta})=\left(G_{d1}(\boldsymbol{\zeta}), \cdots, G_{dq}(\boldsymbol{\zeta}) \right) \subset \mathbb{R}^{1\times \mathtt{q}}$. If $G_{d\mathfrak{i}}(\boldsymbol{\zeta}) \subset \mathbb{R}$ for $\mathfrak{i} \in \{1,\ldots,\mathtt{q}\}$ and there are symbol changes over the interval $[t_\mathtt{0}, t_\mathtt{z}]$, it is necessary to reconsider whether Theorem \ref{theorem2} holds for the following two cases. \par 
	Define $\boldsymbol{\mathfrak{\upsilon}} \!=\! (\mathtt{\upsilon}_{1},\cdots,\mathtt{\upsilon}_{\mathtt{q}})^{T}$, $\! \boldsymbol{\mathfrak{\upsilon}}_{\min} \!=\! (\mathtt{\upsilon}_{1,\min},\cdots,\mathtt{\upsilon}_{\mathtt{q},\min})^{T} \!\leq\! \boldsymbol{0}$, $\boldsymbol{\mathfrak{\upsilon}}_{\max} = (\mathtt{\upsilon}_{1,\max},\cdots,\mathtt{\upsilon}_{\mathtt{q},\max})^{T} \geq \boldsymbol{0}$. \par 
	\textbf{Case 1:} The control constraints for $\mathtt{\upsilon}_{\mathfrak{i}}, \mathfrak{i} \in \{1,\ldots,\mathtt{q}\}$ exhibit symmetry, that is, $\mathtt{\upsilon}_{\mathfrak{i},\max}=-\mathtt{\upsilon}_{\mathfrak{i},\min}$. Under the circumstance, if $G_{d\mathfrak{i}}(\boldsymbol{\zeta})>0$, we can conclude that
	\begin{equation}\label{19}
	G_{d\mathfrak{i}}(\boldsymbol{\zeta}) \mathtt{\upsilon}_{\mathfrak{i}, \min } \leq G_{d\mathfrak{i}}(\boldsymbol{\zeta}) \mathtt{\upsilon}_{\mathfrak{i}} \leq G_{d\mathfrak{i}}(\boldsymbol{\zeta}) \mathtt{\upsilon}_{\mathfrak{i}, \max }.
	\end{equation} \par
	If $G_{d\mathfrak{i}}(\boldsymbol{\zeta})$ changes the symbol at some time $t_{c} \in [t_{\mathtt{0}}, t_{\mathtt{z}}]$ and $G_{d\mathfrak{j}}(\boldsymbol{\zeta}) (\mathfrak{j} \neq \mathfrak{i})$ remains unchanged, then
	\begin{equation}\label{20}
	\begin{aligned}
	&G_{d \mathfrak{i}}(\boldsymbol{\zeta}) \mathtt{\upsilon}_{\mathfrak{i}, \max } \leq G_{d \mathfrak{i}}(\boldsymbol{\zeta}) \mathtt{\upsilon}_\mathfrak{i} \leq G_{d \mathfrak{i}}(\boldsymbol{\zeta}) \mathtt{\upsilon}_{\mathfrak{i}, \min }, \\
	&G_{d \mathfrak{j}}(\boldsymbol{\zeta}) \mathtt{\upsilon}_{\mathfrak{j}, \min } \leq G_{d \mathfrak{j}}(\boldsymbol{\zeta}) \mathtt{\upsilon}_j \leq G_{d \mathfrak{j}}(\boldsymbol{\zeta}) \mathtt{\upsilon}_{\mathfrak{j}, \max }.
	\end{aligned}
	\end{equation}
	According to \eqref{20},
	\begin{equation}\label{21}
	\left|G_{d}(\boldsymbol{\zeta})\right| \boldsymbol{\mathfrak{\upsilon}}_{\min } \leq G_{d}(\boldsymbol{\zeta}) \boldsymbol{\mathfrak{\upsilon}} \leq\left|G_{d}(\boldsymbol{\zeta})\right| \boldsymbol{\mathfrak{\upsilon}}_{\max }.
	\end{equation}
	
	The sufficient and necessary condition under which \eqref{10} does not conflict with \eqref{21} is
	\begin{equation}\label{11}
	\begin{aligned}
	\frac{(\mathtt{\Theta}+1) \mathtt{\theta}}{(\|\boldsymbol{\zeta}\!+\!\boldsymbol{\mathfrak{d}}\| \!+ \! \mathtt{r})^2\|\boldsymbol{\zeta} \! + \! \boldsymbol{\mathfrak{d}}\|}[\left|G_d\right| \boldsymbol{\mathfrak{\upsilon}}_{\min } \! + \!(\boldsymbol{\zeta} \! + \! \boldsymbol{\mathfrak{d}})^T \mathtt{\hat{F}}(\boldsymbol{\zeta})-\\
	\frac{(\|\boldsymbol{\zeta}+\boldsymbol{\mathfrak{d}}\|+\mathtt{r})\|\boldsymbol{\zeta}+\boldsymbol{\mathfrak{d}}\| L_\mathtt{\hat{F}} \mathtt{\theta}}{\mathtt{\theta}}] \leq \alpha(\mathtt{\Theta}).
	\end{aligned}
	\end{equation}
	
	\begin{figure}[!htbp]
		\centering
		\includegraphics[width=3in,angle=0]{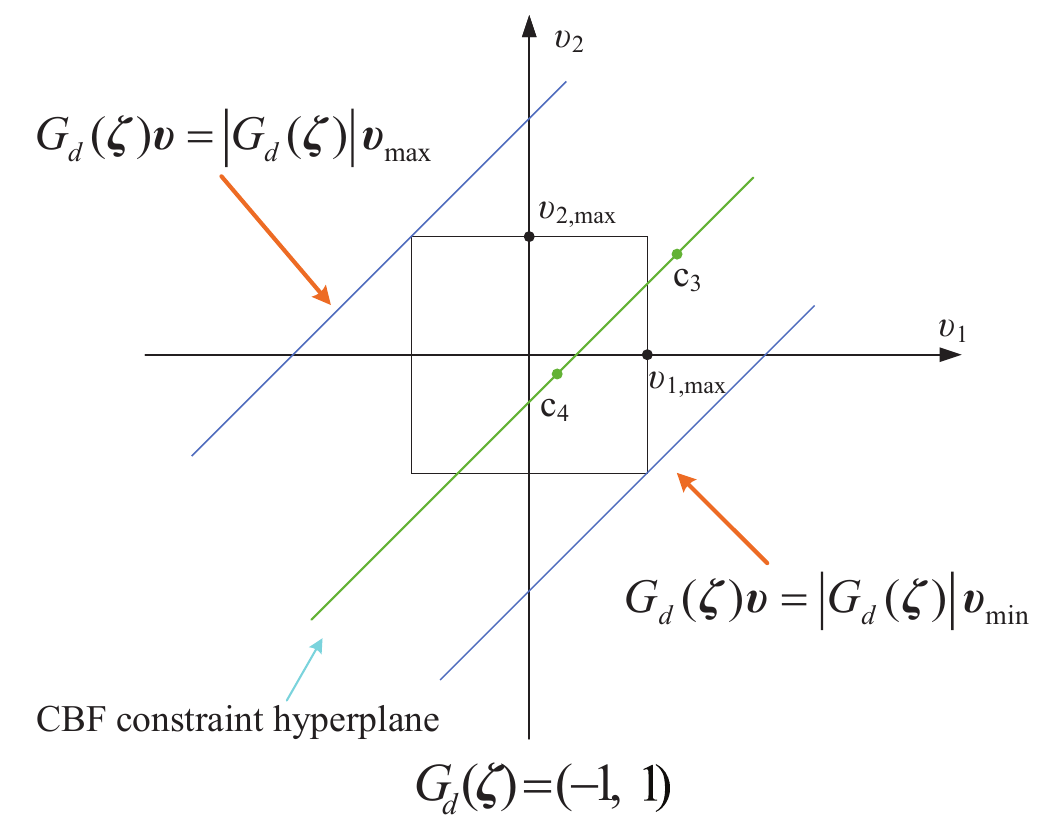}
		\captionsetup{font={small}}
		\caption{The correlation between the confinement of control input \eqref{45} and the constraint of Control Barrier Function (CBF) \eqref{44} in a scenario where a two-dimensional vector $\boldsymbol{\mathfrak{\upsilon}}=(\mathtt{\upsilon}_1, \mathtt{\upsilon}_2)^{T}$ is employed, and $G_{d\mathfrak{i}}(\boldsymbol{\zeta})$ modifies the symbol. The inclination of the two blue lines or hyperplanes $G_d(\boldsymbol{\zeta}) \boldsymbol{\mathfrak{\upsilon}} = \left|G_{d}(\boldsymbol{\zeta})\right| \boldsymbol{\mathfrak{\upsilon}}_{\min}$ and $G_d(\boldsymbol{\zeta}) \boldsymbol{\mathfrak{\upsilon}} = \left|G_{d}(\boldsymbol{\zeta})\right| \boldsymbol{\mathfrak{\upsilon}}_{\max}$ is sensitive to the magnitude of $G_{d}(\boldsymbol{\zeta})$.}
		\label{fig5}
	\end{figure} \par
	From Fig. \ref{fig0} and \ref{fig5}, \eqref{11} is the sufficient and necessary condition no matter how the symbol of $G_{d}(\boldsymbol{\zeta})$ changes in $[t_\mathtt{0}, t_\mathtt{z}]$.
	
	\textbf{Case 2:} The control constraints for $\mathtt{\upsilon}_{\mathfrak{i}}, \mathfrak{i} \in \{1,\cdots,\mathtt{q}\}$ exhibit asymmetry, that is, $\mathtt{\upsilon}_{\mathfrak{i},\max}\neq-\mathtt{\upsilon}_{\mathfrak{i}, \min}$. Under the circumstance, let
	\begin{equation}\label{22}
	\mathtt{\upsilon}_{\mathfrak{i}, \lim }:=\min \left\{\left|\mathtt{\upsilon}_{\mathfrak{i}, \min }\right|, \mathtt{\upsilon}_{\mathfrak{i}, \max }\right\}.
	\end{equation}  \par
	Let $\boldsymbol{\underline \upsilon}_{\min} = -(\mathtt{\upsilon}_{1,\lim},\cdots,\mathtt{\upsilon}_{\mathtt{q},\lim})^{T}$, $\boldsymbol{\overline \upsilon}_{\max} = (\mathtt{\upsilon}_{1,\lim},\cdots,\mathtt{\upsilon}_{\mathtt{q},\lim})^{T}$. Thus, \eqref{45} can be rewritten as
	\begin{equation}\label{24}
	\boldsymbol{\underline \upsilon}_{\min} \leq \boldsymbol{\mathfrak{\upsilon}} \leq \boldsymbol{\overline \upsilon}_{\max}.
	\end{equation} \par 
	It is evident that \textbf{Case 2} can be processed similar to \textbf{Case 1}. Therefore, a change in sign of $G_{d\mathfrak{i}}(\boldsymbol{\zeta}), \mathfrak{i} \in \{1,\ldots,\mathtt{q}\}$ does not affect Theorem \ref{theorem2}. \par
	To summarize, we can determine that \eqref{11} represents a crucial and comprehensive prerequisite for the control tactics to attain feasibility.
\par 

\subsection{}\label{B}
\textit{Proof of \textbf{Theorem \ref{theorem3}}: }
	It is easy to know that $Y(\boldsymbol{\zeta})$ is a continuous function. Then, we only need to determine whether there exists $Y(\boldsymbol{\zeta}) \rightarrow +\infty$ when $\mathtt{\Theta}(\boldsymbol{\zeta}) \rightarrow 0$. \par 
	
	The premise for system \eqref{1} to reach stability is that starting from a specific moment $t_s$, the value of $\boldsymbol{\mathfrak{\upsilon}}(t)$ satisfies 
	\begin{equation}\label{32}
	\dot{\mathtt{\Theta}}(\boldsymbol{\zeta}) + \alpha{(\mathtt{\Theta}(\boldsymbol{\zeta}))} = 0.
	\end{equation} \par 
	We selected $ K_{\mathtt{\Theta}} $ ($K_{\mathtt{\Theta}}$ is a positive constant). We constructed a Lyapunov function $\mathtt{V}_e(\boldsymbol{\zeta})=\frac{1}{2}(\boldsymbol{\zeta}+\boldsymbol{\mathfrak{d}})^T(\boldsymbol{\zeta}+\boldsymbol{\mathfrak{d}})+L \mathtt{\Theta}(\boldsymbol{\zeta})$, where $L$ is a positive constant. The derivative of $\mathtt{V}_e$ with respect to time $t$ is
	\begin{equation}\label{18}
	\begin{gathered}
	\dot{\mathtt{V}}_e(\boldsymbol{\zeta})=\frac{(\|\boldsymbol{\zeta}+\boldsymbol{\mathfrak{d}}\|+\mathtt{r})\|\boldsymbol{\zeta}+\boldsymbol{\mathfrak{d}}\| L_\mathtt{\hat{F}} \mathtt{\theta}}{\mathtt{\theta}}+ \\
	\hspace{2cm} ~~\left(\!\frac{(\|\boldsymbol{\zeta}\!+\!\boldsymbol{\mathfrak{d}}\|\!+\!\mathtt{r})^2\|\boldsymbol{\zeta}\!+\!\boldsymbol{\mathfrak{d}}\|}{(\mathtt{\Theta}\!+\!1) \mathtt{\theta}}\!-\!L\right) \alpha(\mathtt{\Theta}).
	\end{gathered}
	\end{equation} \par 
	We select $L \geq \max \left(\frac{(\|\boldsymbol{\zeta}+\boldsymbol{\mathfrak{d}}\|+\mathtt{r})^2\|\boldsymbol{\zeta}+\boldsymbol{\mathfrak{d}}\|}{(\mathtt{\Theta}+1) \mathtt{\theta}}\right)$. Because $\dot{\mathtt{\theta}}<0  \; (\|\boldsymbol{\zeta}-\boldsymbol{\zeta}_{\mathtt{d}}\| < \eta, 0<\eta<\epsilon)$, then, $\dot{\mathtt{V}}_e(\boldsymbol{\zeta}) < 0 (\mathtt{\Theta}>0)$. $\|\boldsymbol{\zeta}+\boldsymbol{\mathfrak{d}}\| \neq 0$ for Assumption \ref{assumption2}, we know that $\dot{\mathtt{\theta}} = 0, \mathtt{\Theta} = 0$ and $\dot{\mathtt{V}}_e(\boldsymbol{\zeta}) = 0$ when $t \rightarrow \infty$; that is, system \eqref{1} eventually remains stable, $\dot{\boldsymbol{\zeta}} = \boldsymbol{0}$. \par 
	According to the dynamic system \eqref{1}, in a stable state, $G_d(\boldsymbol{\zeta}) \boldsymbol{\mathfrak{\upsilon}} = -(\boldsymbol{\zeta}+\boldsymbol{\mathfrak{d}})^T \mathtt{\hat{F}}(\boldsymbol{\zeta})$. As $(\boldsymbol{\zeta}+\boldsymbol{\mathfrak{d}})^T \mathtt{\hat{F}}(\boldsymbol{\zeta})+\left|G_d(\boldsymbol{\zeta})\right| \boldsymbol{\mathfrak{\upsilon}}_{\min } \leq(\boldsymbol{\zeta}+\boldsymbol{\mathfrak{d}})^T \mathtt{\hat{F}}(\boldsymbol{\zeta})+G_d(\boldsymbol{\zeta}) \boldsymbol{\mathfrak{\upsilon}}=0$ when $\dot{\boldsymbol{\zeta}} = \boldsymbol{0}$. Thus, $\underset{\mathtt{\Theta}\rightarrow0}{\lim}Y(\boldsymbol{\zeta}) \leq 0$. \par 
	In summary, if the design method can ensure that \eqref{32} is established at a particular moment, then $\exists N>0$, $\underset{\boldsymbol{\zeta} \in \mathtt{\bar{C}}}{\max} \{Y(\boldsymbol{\zeta})\} \leq N$; that is, $\exists N>0, \forall \boldsymbol{\zeta} \in \mathtt{\bar{C}}$ such that $(\mathtt{\Theta}+1)Y(\boldsymbol{\zeta}) \leq (\mathtt{\Theta}+1)N < +\infty$. Thus, there exists a class of $\mathcal{K}$ functions that satisfy \eqref{11}.

\bibliographystyle{Bibliography/IEEEtranTIE}
\bibliography{Bibliography/IEEEabrv,Bibliography/BIB_CBF}\ %IEEEabrv instead of IEEEfull

% Generated by IEEEtran.bst, version: 1.12 (2007/01/11)
\begin{thebibliography}{10}
\providecommand{\url}[1]{#1}
\csname url@samestyle\endcsname
\providecommand{\newblock}{\relax}
\providecommand{\bibinfo}[2]{#2}
\providecommand{\BIBentrySTDinterwordspacing}{\spaceskip=0pt\relax}
\providecommand{\BIBentryALTinterwordstretchfactor}{4}
\providecommand{\BIBentryALTinterwordspacing}{\spaceskip=\fontdimen2\font plus
\BIBentryALTinterwordstretchfactor\fontdimen3\font minus
  \fontdimen4\font\relax}
\providecommand{\BIBforeignlanguage}[2]{{%
\expandafter\ifx\csname l@#1\endcsname\relax
\typeout{** WARNING: IEEEtran.bst: No hyphenation pattern has been}%
\typeout{** loaded for the language `#1'. Using the pattern for}%
\typeout{** the default language instead.}%
\else
\language=\csname l@#1\endcsname
\fi
#2}}
\providecommand{\BIBdecl}{\relax}
\BIBdecl

\bibitem{ding2022security}
F.~Ding, H.~Shan, X.~Han, C.~Jiang, C.~Peng, and J.~Liu, ``Security-based
  resilient triggered output feedback lane keeping control for human-machine
  cooperative steering intelligent heavy truck under denial-of-service
  attacks,'' \emph{IEEE Transactions on Fuzzy Systems}, 2022.

\bibitem{wei2019risk}
C.~Wei, R.~Romano, N.~Merat, Y.~Wang, C.~Hu, H.~Taghavifar, F.~Hajiseyedjavadi,
  and E.~R. Boer, ``Risk-based autonomous vehicle motion control with
  considering human driver's behaviour,'' \emph{Transportation research part C:
  emerging technologies}, vol. 107, pp. 1--14, 2019.

\bibitem{crosato2022interaction}
L.~Crosato, H.~P. Shum, E.~S. Ho, and C.~Wei, ``Interaction-aware
  decision-making for automated vehicles using social value orientation,''
  \emph{IEEE Transactions on Intelligent Vehicles}, vol.~8, no.~2, pp.
  1339--1349, 2022.

\bibitem{zhang2022collaborative}
J.~Zhang, J.~Liu, and F.~Ding, ``Collaborative optimization design framework
  for hierarchical filter barrier control suspension system with projection
  adaptive tracking hydraulic actuator,'' \emph{Nonlinear Dynamics}, vol. 108,
  no.~4, pp. 3417--3434, 2022.

\bibitem{XuGrizzle2017}
X.~Xu, J.~W. Grizzle, P.~Tabuada, and A.~D. Ames, ``Correctness guarantees for
  the composition of lane keeping and adaptive cruise control,'' \emph{IEEE
  Transactions on Automation Science and Engineering}, vol.~15, no.~3, pp.
  1216--1229, Jul. 2018.

\bibitem{Zhu2022}
G.~Zhu, H.~Li, X.~Zhang, C.~Wang, C.-Y. Su, and J.~Hu, ``Adaptive consensus
  quantized control for a class of high-order nonlinear multi-agent systems
  with input hysteresis and full state constraints,'' \emph{IEEE/CAA Journal of
  Automatica Sinica}, vol.~9, no.~9, pp. 1574--1589, Sep. 2022.

\bibitem{Brudigam2023}
T.~Br\"{u}digam, M.~Olbrich, D.~Wollherr, and M.~Leibold, ``Stochastic model
  predictive control with a safety guarantee for automated driving,''
  \emph{IEEE Transactions on Intelligent Vehicles}, vol.~8, no.~1, pp. 22--36,
  2023.

\bibitem{HJWang2022}
H.~Wang, J.~Peng, F.~Zhang, H.~Zhang, and Y.~Wang, ``High-order control barrier
  functions-based impedance control of a robotic manipulator with time-varying
  output constraints,'' \emph{ISA Transactions}, vol. 129, pp. 361--369, Oct.
  2022.

\bibitem{TaylorAmes2020}
A.~J. Taylor and A.~D. Ames, ``Adaptive safety with control barrier
  functions,'' in \emph{2020 American Control Conference (ACC)}, pp.
  1399--1405, Jul. 2020.

\bibitem{Pickem2017}
D.~Pickem, P.~Glotfelter, L.~Wang, M.~Mote, A.~Ames, E.~Feron, and
  M.~Egerstedt, ``The robotarium: A remotely accessible swarm robotics research
  testbed,'' in \emph{2017 IEEE International Conference on Robotics and
  Automation (ICRA)}, pp. 1699--1706, Jun. 2017.

\bibitem{XuWaters2017}
X.~Xu, T.~Waters, D.~Pickem, P.~Glotfelter, M.~Egerstedt, P.~Tabuada, J.~W.
  Grizzle, and A.~D. Ames, ``Realizing simultaneous lane keeping and adaptive
  speed regulation on accessible mobile robot testbeds,'' in \emph{2017 IEEE
  Conference on Control Technology and Applications (CCTA)}, pp. 1769--1775,
  2017.

\bibitem{WangAmes2017ICRA}
L.~Wang, A.~D. Ames, and M.~Egerstedt, ``Safe certificate-based maneuvers for
  teams of quadrotors using differential flatness,'' in \emph{2017 IEEE
  International Conference on Robotics and Automation (ICRA)}, pp. 3293--3298,
  Jun. 2017.

\bibitem{ZhangY2022}
Y.~Zhang, M.~Xu, Y.~Qin, M.~Dong, L.~Gao, and E.~Hashemi, ``Mile:
  Multi-objective integrated model predictive adaptive cruise control for
  intelligent vehicle,'' \emph{IEEE Transactions on Industrial Informatics},
  pp. 1--9, 2022.

\bibitem{AmesGalloway2014}
A.~D. Ames, K.~Galloway, K.~Sreenath, and J.~W. Grizzle, ``Rapidly
  exponentially stabilizing control lyapunov functions and hybrid zero
  dynamics,'' \emph{IEEE Transactions on Automatic Control}, vol.~59, no.~4,
  pp. 876--891, Apr. 2014.

\bibitem{Gangopadhyay2022}
B.~Gangopadhyay, P.~Dasgupta, and S.~Dey, ``Safe and stable {RL (S2RL)} driving
  policies using control barrier and control lyapunov functions,'' \emph{IEEE
  Transactions on Intelligent Vehicles}, pp. 1--1, Feb. 2022.

\bibitem{HuWang2021}
C.~Hu and J.~Wang, ``Trust-based and individualizable adaptive cruise control
  using control barrier function approach with prescribed performance,''
  \emph{IEEE Transactions on Intelligent Transportation Systems}, vol.~23,
  no.~7, pp. 6974--6984, Jul. 2022.

\bibitem{AmesXu2017}
A.~D. Ames, X.~Xu, J.~W. Grizzle, and P.~Tabuada, ``Control barrier function
  based quadratic programs for safety critical systems,'' \emph{IEEE
  Transactions on Automatic Control}, vol.~62, no.~8, pp. 3861--3876, Aug.
  2017.

\bibitem{ZengCBF2021}
J.~Zeng, B.~Zhang, Z.~Li, and K.~Sreenath, ``Safety-critical control using
  optimal-decay control barrier function with guaranteed point-wise
  feasibility,'' in \emph{2021 American Control Conference (ACC)}, pp.
  3856--3863, May. 2021.

\bibitem{Khalil2002}
P.~G. Drazin and P.~D. Drazin, \emph{Nonlinear systems}, 3rd~ed., no.~10.\hskip
  1em plus 0.5em minus 0.4em\relax USA: Cambridge University Press, 1992.

\bibitem{Xiao2019}
W.~Xiao and C.~Belta, ``Control barrier functions for systems with high
  relative degree,'' in \emph{2019 IEEE 58th Conference on Decision and Control
  (CDC)}, pp. 474--479, Dec. 2019.

\bibitem{Hussain2018}
M.~Hussain, M.~Rehan, C.~Ki~Ahn, and M.~Tufail, ``Robust antiwindup for
  one-sided lipschitz systems subject to input saturation and applications,''
  \emph{IEEE Transactions on Industrial Electronics}, vol.~65, no.~12, pp.
  9706--9716, Dec. 2018.

\bibitem{NguyenLaurain2016}
A.-T. Nguyen, T.~Laurain, R.~Palhares, J.~Lauber, C.~Sentouh, and J.-C.
  Popieul, ``{LMI}-based control synthesis of constrained takagi-sugeno fuzzy
  systems subject to $\mathscr{L}_{2}$ or $\mathscr{L}_{\infty}$
  disturbances,'' \emph{Neurocomputing}, vol. 207, pp. 793--804, Sep. 2016.

\bibitem{Zhang2021}
J.~Zhang, K.~Li, and Y.~Li, ``Output-feedback based simplified optimized
  backstepping control for strict-feedback systems with input and state
  constraints,'' \emph{IEEE/CAA Journal of Automatica Sinica}, vol.~8, no.~6,
  pp. 1119--1132, Jun. 2021.

\bibitem{Nguyen2021}
A.-T. Nguyen, P.~Coutinho, T.-M. Guerra, R.~Palhares, and J.~Pan, ``Constrained
  output-feedback control for discrete-time fuzzy systems with local nonlinear
  models subject to state and input constraints,'' \emph{IEEE Transactions on
  Cybernetics}, vol.~51, no.~9, pp. 4673--4684, Sep. 2021.

\bibitem{Xiao2022}
W.~Xiao, C.~A. Belta, and C.~G. Cassandras, ``Sufficient conditions for
  feasibility of optimal control problems using control barrier functions,''
  \emph{Automatica}, vol. 135, p. 109960, Jan. 2022.

\bibitem{Glotfelter2017}
P.~Glotfelter, J.~Cortés, and M.~Egerstedt, ``Nonsmooth barrier functions with
  applications to multi-robot systems,'' \emph{IEEE Control Systems Letters},
  vol.~1, no.~2, pp. 310--315, Jun. 2017.

\bibitem{Lindemann2019}
L.~Lindemann and D.~V. Dimarogonas, ``Control barrier functions for signal
  temporal logic tasks,'' \emph{IEEE Control Systems Letters}, vol.~3, no.~1,
  pp. 96--101, Jul. 2019.

\bibitem{Nocedal2006}
\BIBentryALTinterwordspacing
S.~W. Jorge~Nocedal, \emph{Numerical Optimization}, 2nd~ed., ser. Springer
  series in operations research.\hskip 1em plus 0.5em minus 0.4em\relax Berlin,
  Germany: Springer, 2006. [Online]. Available:
  \url{http://gen.lib.rus.ec/book/index.php?md5=7016b74cfe6dc64c75864322ee4aa081}
\BIBentrySTDinterwordspacing

\bibitem{Payman2012}
P.~Shakouri, A.~Ordys, and M.~R. Askari, ``Adaptive cruise control with
  stop$\&$go function using the state-dependent nonlinear model predictive
  control approach,'' \emph{ISA Transactions}, vol.~51, no.~5, pp. 622--631,
  Sep. 2012.

\bibitem{Ioannou1993}
P.~Ioannou and C.~Chien, ``Autonomous intelligent cruise control,'' \emph{IEEE
  Transactions on Vehicular Technology}, vol.~42, no.~4, pp. 657--672, Nov.
  1993.

\end{thebibliography}

\vspace{-2.5cm}
\begin{IEEEbiography}[{\includegraphics[width=1in,height=1.25in,clip,keepaspectratio]{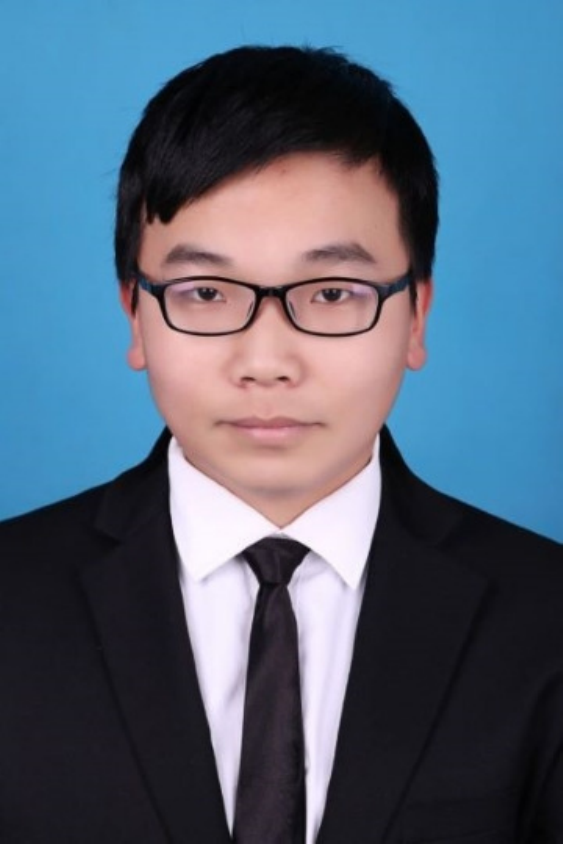}}]
	{Fan Yang} received the B.S. degree in automation from Wuhan Institute of Technology, Wuhan, China, in 2019. He is currently pursuing the M.S. degree in control science and engineering with the School of Automation Engineering, University of Electronic Science and Technology of China, Chengdu, China. His current research interests include state-constrained control and autonomous driving. \par
\end{IEEEbiography}

\vspace{-2cm}
\begin{IEEEbiography}[{\includegraphics[width=1in,height=1.25in,clip,keepaspectratio]{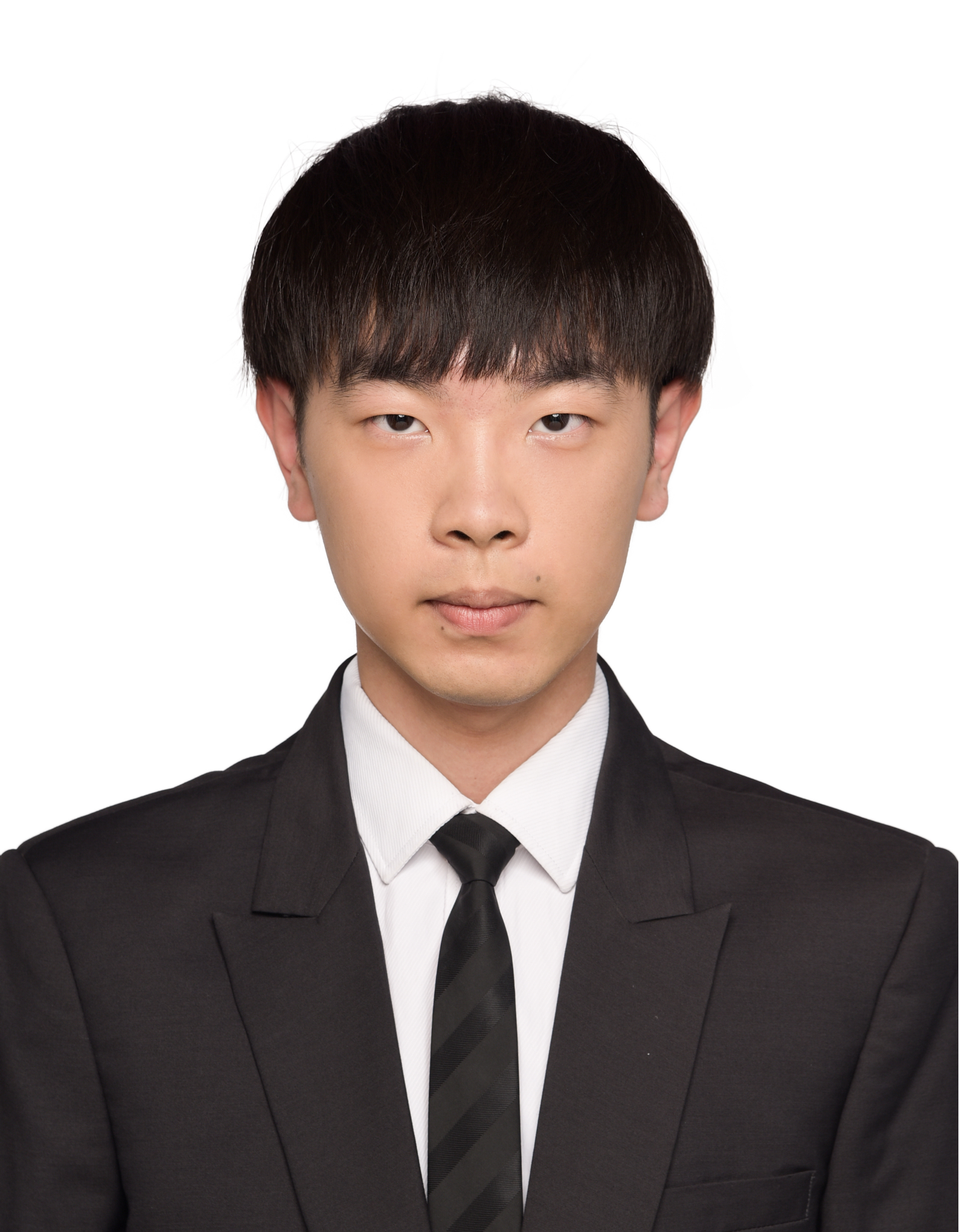}}]
	{Haoqi Li} received the B.S. degree and the M.S. degree from the Shanghai University of Electric Power and the Northeast Electric Power University in 2018 and 2022, respectively. He is currently pursuing the Ph.D. degree in control science and engineering from the School of Automation Engineering, University of Electronic Science and Technology of China. His research interests include adaptive and state constrained control for multi-agent systems and nonlinear systems.
	%\\ \\ 
\end{IEEEbiography}

%\vspace{-2cm}
\begin{IEEEbiography}[{\includegraphics[width=1in,height=1.25in,clip,keepaspectratio]{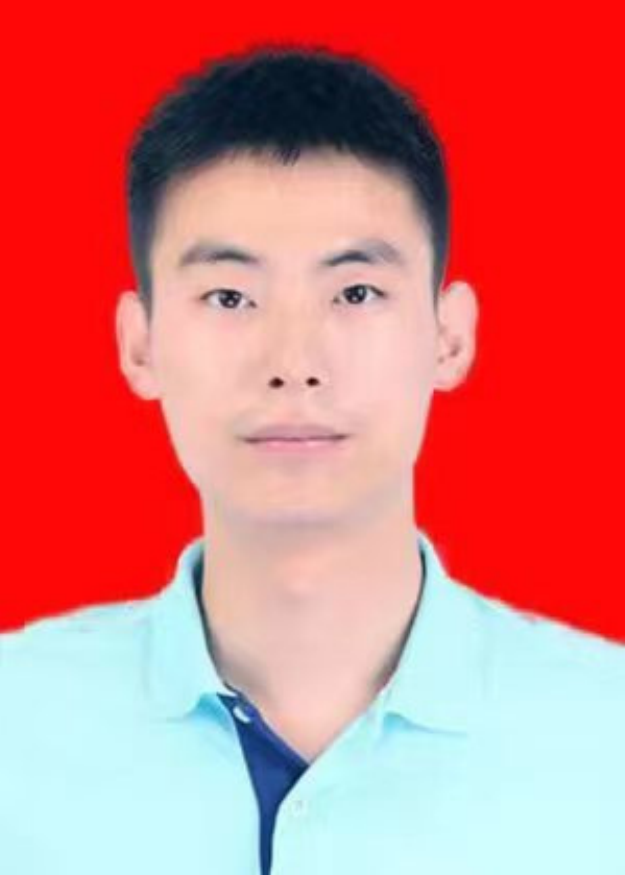}}]
	{Maolong Lv} received the Ph.D. degree from Delft Center for Systems and Control, Delft University of Technology, Delft, The Netherlands, in 2021. \par
	He is currently with Air Force Engineering University, Xi'an, China. His research interests include adaptive learning control, distributed control, reinforce learning, and intelligent decision making with applications in multiagent systems, hypersonic vehicles, and unmanned autonomous systems. \par
	Dr. Lv was awarded with Descartes Excellence Fellowship from the French Goverment in 2018, which allowed him a Research Visit from 2018 to 2019 at University of Grenoble working on adaptive networked systems with emphasis on traffic with human deriven and autonomous vehicles. He was awarded the Yong Talent Support Project for Military Science and Technology, the Yong Talent Fund of Association for Science and Technology in Shanxi, and the Postdoctoral International Exchange Program in 2022. He is currently an Editor for Aerospace and Measurement and Control. 
	%\\ \\ 
\end{IEEEbiography}

\vspace{-1cm}
\begin{IEEEbiography}[{\includegraphics[width=1in,height=1.25in,clip,keepaspectratio]{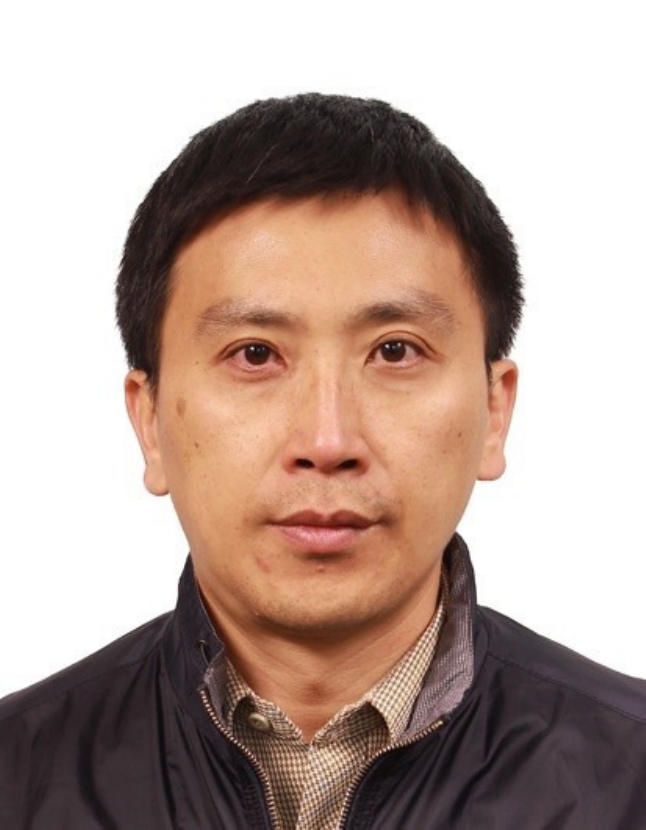}}]
	{Jiangping Hu} (Senior Member, IEEE) received the B.S. degree in applied mathematics and the M.S. degree in computational mathematics from Lanzhou University, Lanzhou, China, in 2000 and 2004, respectively, and the Ph.D. degree in modelling and control of complex systems from the Academy of Mathematics and Systems Science, Chinese Academy of Sciences, Beijing, China, in 2007. He has held various positions with the Royal Institute of Technology, Stockholm, Sweden, The City University of Hong Kong, Hong Kong, Sophia University, Tokyo, Japan, and Western Sydney University, Sydney, NSW, Australia. \par 
	He is currently a Professor with the School of Automation Engineering, University of Electronic Science and Technology of China, Chengdu, China. His current research interests include multi-agent systems, social dynamics, and sensor networks. \par 
	Dr. Hu has served as an Associate Editor for KYBERNETIKA and an Associate Editor for Journal of Systems Science and Complexity.
\end{IEEEbiography}

%\vspace{-1cm}
\begin{IEEEbiography}[{\includegraphics[width=1in,height=1.25in,clip,keepaspectratio]{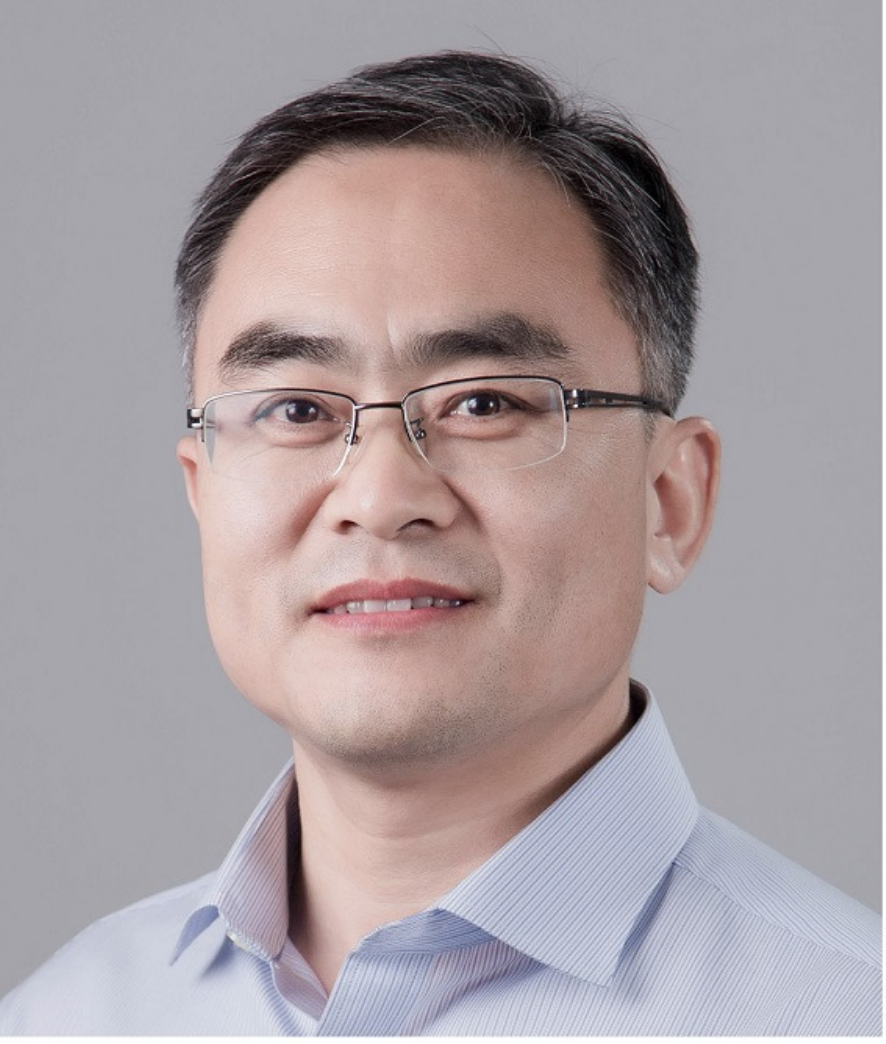}}]
	{Qingrui Zhou}  received the B.S. degree in control science and the M.S. degree in Operational Research and Cybernetics from Shandong University,  in 1994 and 2002, respectively, and the Ph.D. degree in Control Theory and Control Engineering from the Institute of Automation, Chinese Academy of Sciences, Beijing, China, in 2005.  \par 
	He is currently a Professor with the Qianxuesen Lab, China Academy of Space Technology, Beijing, China. His current research interests include spacecraft formation flying, spacecraft navigation and control, swarm intelligence. 	
\end{IEEEbiography}

\vspace{2.5cm}
\begin{IEEEbiography}[{\includegraphics[width=1in,height=1.25in,clip,keepaspectratio]{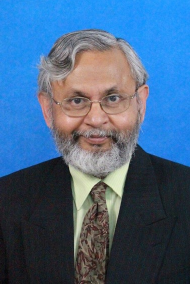}}]
	{Bijoy Kumar Ghosh} (Life Fellow, IEEE) received the Ph.D. degree from Harvard University in 1983. From 1983 to 2007, he was with the Department of Electrical and Systems Engineering, Washington University. He is currently the Dick and Martha
	Brooks Regents Professor of Mathematics and Statistics at Texas Tech University, Lubbock, TX, USA. His research interests include biomechanics, cyberphysical systems, and control problems in rehabilitation engineering. Dr. Ghosh became a fellow of the International
	Federation on Automatic Control in 2014.
	%\\ \\ 
\end{IEEEbiography}

\end{document}